\documentclass[preprint, 3p]{elsarticle}

\usepackage[english]{babel}
\usepackage[utf8]{inputenc}
\usepackage{amsthm}
\usepackage{amsmath}
\usepackage{amssymb}
\usepackage{hyperref}
\usepackage{enumerate}
\usepackage{color,soul}
\usepackage{tikz}
\usepackage{placeins}
\usepackage{subfigure} 
\usepackage{booktabs}
\usepackage{bbm}
\usepackage[ruled,vlined]{algorithm2e}
\usepackage{appendix}
\usepackage{lscape}
\usepackage{longtable}
\usepackage[normalem]{ulem}
\useunder{\uline}{\ul}{}

\journal{TBA}

\newtheorem{theorem}{Theorem}[section]
\newtheorem{corollary}[theorem]{Corollary}
\newtheorem{lemma}[theorem]{Lemma}
\newtheorem{proposition}[theorem]{Proposition}
\newtheorem{definition}[theorem]{Definition}
\newtheorem{remark}[theorem]{Remark}
\newtheorem{example}[theorem]{Example}
\numberwithin{equation}{section}

\renewcommand{\appendix}{\par
  \setcounter{section}{0}
  \setcounter{subsection}{0}
  \gdef\thesection{\Alph{section}}
}

\DeclareMathOperator{\sgn}{sgn}
\DeclareMathOperator{\tr}{tr}

\DeclareMathOperator{\argmin}{argmin}

\begin{document}
	
\begin{frontmatter}
	
\title{Approximately optimal trade execution strategies under fast mean-reversion}

\author[FGV]{D. Evangelista}
\ead{david.evangelista@fgv.br}

\author[IMEUFF]{Y. Thamsten}
\ead{ythamsten@id.uff.br}

\address[FGV]{Escola de Matem\'atica Aplicada (EMAp), Funda\c{c}\~ao Get\'ulio Vargas (FGV), 22250-900, Rio de Janeiro, RJ, Brasil}

\address[IMEUFF]{Instituto de Matem\'atica e Estat\'istica (IME), Universidade Federal Fluminense (UFF), 24210-380, Niter\'{o}i, RJ, Brasil}
	
\date{}
	
%\maketitle
	
\begin{abstract}
    In a fixed time horizon, appropriately executing a large amount of a particular asset --- meaning a considerable portion of the volume traded within this frame --- is challenging. Especially for illiquid or even highly liquid but also highly volatile ones, the role of ``market quality'' is quite relevant in properly designing execution strategies. Here, we model it by considering uncertain volatility and liquidity; hence, moments of high or low price impact and risk vary randomly throughout the trading period. We work under the central assumption: although there are these uncertain variations, we assume they occur in a fast mean-reverting fashion. We thus employ singular perturbation arguments to study approximations to the optimal strategies in this framework. By using high-frequency data, we provide estimation methods for our model in face of microstructure noise, as well as numerically assess all of our results.
\end{abstract}	


\end{frontmatter}	

\section{Introduction} \label{sec:intro}

\subsection{The optimal execution problem}

Whenever we want to execute (meaning to liquidate or to acquire) a large volume of a particular asset, several difficulties arise. Basically, the bulk of execution algorithms is to deal, in the most excellent way possible, with the trade-off between two financial complexities: trading costs and the uncertainty in price movements. On the one hand, we manage the latter aspect by fixing a trading horizon and then fractionating the larger trade into smaller ones, i.e., by setting up a trading schedule. On the other hand, some direct trading costs, such as fees from brokerage firms, are easy to handle. However, the same is not valid for some indirect costs --- it is sometimes hard for us to even describe the latter qualitatively, let alone quantify them. Here, we will concentrate on the type of indirect cost known as price impact. This issue is relatively deep, being the subject of many papers in both the empirical and theoretical literature. Early efforts in modeling price impact are \cite{glosten1985bid,kyle1985continuous} --- see also the efforts \cite{back2004information,krishnan1992equivalence} of linking those two. Relevant empirical advances comprise \cite{almgren2005direct,cont2014price}; see also \cite{bouchaud2010price,cartea2015algorithmic,laruelle2018market} for more in-depth discussions on market microstructure and price impact.

A ubiquitous type of impact in the optimal execution literature is the \textit{temporary} one. This market friction represents the excess amount a trader has to pay per share, relative to the marked-to-market price, to consume additional layers of the Limit Order Book (LOB) to have an order she sent filled. Under the assumption that the temporary price impact is linear on the agent's turnover rate, Almgren and Chriss (AC) proposed in the seminal work \cite{almgren2001optimal} their celebrated model. They also considered a \textit{permanent} price impact, in a way to capture the influence of the trades on the dynamics of the asset's price. Further important advances regarding price impacts in optimal execution include the model of Bertsimas and Lo \cite{bertsimas1998optimal} (which is, in a sense, a precursor of the AC model) and the model of Obizhaeva and Wang \cite{obizhaeva2013optimal}, modeling the LOB using supply/demand functions.

It is fair to say that the AC approach to the execution problem led to a flourishing of the field, especially in what regards generalizations of their model. The work \cite{gatheral2011optimal} concerns the use of a geometric Brownian motion instead of an arithmetic one. The papers \cite{cartea2015optimal,gueant2012optimal} address the use of limit orders simultaneously to market orders. Regarding price impacts, a transient type of this cost figures in \cite{gatheral2012transient} as an alternative to its permanent counterpart. Also, \cite{almgren2003optimal} contain contributions relaxing the linearity assumption on the temporary and permanent impacts per share. In \cite{cartea2016incorporating}, authors regard the presence of a background noise affecting price dynamics. Another important assumption of the AC model is its measure of execution quality: they benchmark their performance with the pre-trade price, leading to \textit{Implementation Shortfall} orders. We refer to \cite{cartea2015algorithmic,gueant2016financial} for investigations of the execution problem under other benchmarks. Moreover, it is worthwhile to remark that the AC framework was quite suitable for studying problems other than optimal execution, such as hedging \cite{almgren2016option,bank2017hedging,gueant2017option}. Recently, works such as \cite{evangelista2022price,feron2020price,feron2021price,fujii2019probabilistic,fujii2023mean,fujii2021equilibrium,fujii2022equilibrium,fujii2022mean,fujii2022strong} used this frictional market model to investigate the problem of price formation; see also \cite{fu2022portfolio,micheli2021closed,neuman2023trading} for further related game-theoretic models. 

\subsection{Stochastic volatility and liquidity} \label{subsec:StochVolAndLiq}

For large-capitalization stocks in highly liquid markets, it is commonly reasonable to assume that the temporary impact is either constant or has a deterministic profile --- see, e.g., the third panel in \cite[Figure 2]{cartea2016incorporating}. However, a typical issue is that less liquid or highly volatile assets are more difficult to trade. Indeed, as Almgren describes in \cite{almgren2012optimal}, for assets in the former class, there are some moments in the day when trading is cheap and others when negotiating is expensive; at some times, delaying trades is near to cost less, and at others, doing so results in a lot of volatility risk. For some highly volatile assets such as cryptocurrencies, the assumptions of constant volatility and liquidity are also quite far from true. A subtlety that renders the problem even harder is that these circumstances vary randomly throughout the day. Thus, in these scenarios, adaptive trading strategies can be expected to perform better.

The way \cite{almgren2012optimal} models stochastic volatility and liquidity is by assuming that the temporary price impact coefficient and the volatility of the asset price (which in turn follows an arithmetic Brownian motion) are both stochastic. On top of that, here we argue in favor of using some asymptotic techniques that are closely related to the ones that Fouque et al. applied in \cite{fouque2011multiscale} to several financial problems, especially option pricing. For singular approximations to even more general multiscale optimal control problems, M. Bardi et al. made several advances in the last decades, see \cite{alvarez2002viscosity,alvarez2003singular,alvarez2010ergodicity,bardi2010convergence,bardi2011optimal,bardi2014large,bardi2022deep,bardi2023singular}. The basic modeling assumption is that the underlying stochastic processes have as their drivers fast mean-reverting ones; see \cite[Subsection 3.2]{fouque2011multiscale} for detailed discussions in this matter. In the sequel, we will make the case about how we can see stochastic volatility and liquidity as fast mean-reverting. We will do so by arguing empirically. 

\subsection{The data}

Throughout this paper, we will use a data set $\mathcal{D} = \left\{ L_t \right\}_t$ of level two order book updates of the asset BTCUSDT traded on the Binance spot cryptocurrencies exchange.\footnote{For estimating price impact parameters, using proprietary execution data is more adequate, see \cite{almgren2005direct}. However, using public data also leads to reasonable models, cf. \cite{cartea2015algorithmic} for such an approach, which presents results in line with \cite{almgren2005direct}.}\footnote{We remark that crytocurrency markets provide an appropriate setting for the use of models such as we develop. In effect, since the tick size is usually very small, and we can typically trade really small amounts (such as $10^{-5}$ BTC in the current setting), the assumptions of continuous inventory and continuous prices are good approximations to reality. Those markets have been calling the attention of the trade execution research community; we refer to \cite{jaimungal2023optimal} and the references therein for more discussions on the matter.} It contains all the order book states\footnote{Thus, each $L_t \in \mathcal{D}$ is such that $L_t \in \mathbb{R}^{101},$ as they are formed by the timestamp $t$ (making $L_t \neq L_{t^\prime}$ if $t\neq t^\prime$), as well as $25$ order book layers, each of which is a 4-tuple comprised of a bid price, bid amount, ask size, and ask amount. } (up to twenty-five layers for the ask and for the bid) at each time when Binance sends an update.\footnote{When subscribing to Binance's spot market stream channel, Binance sends at most one order book update each 100ms. When there are more than one update within such a time frame, they aggregate all of them in a single message.} Unless we state otherwise, all plots will be relative to December 19, 2022. We provide a few descriptive statistics of this data set in Table \ref{tab:data}, and the full mid-price path in Figure \ref{fig:btcusdt_path_20231219}, alongside some trade data to bring further insight into this market's behaviour. 

% \usepackage{booktabs}
\begin{table}[]
\centering
\begin{tabular}{cc}
Symbol                        & BTCUSDT     \\\midrule
Tick size {[}\${]}            & 0.01        \\\midrule
Bid-ask spread {[}\${]}       & 0.44        \\
                              & (0.27)      \\\midrule
Mid-quote {[}\${]}            & 16676.24    \\
                              & (95.99)     \\\midrule
Number of seconds within      & 0.135       \\
a day per number of LOB       &             \\
updates {[}seconds{]}         &             \\\midrule
Total ask liquidity {[}BTC{]} & 2.0595      \\
                              & (2.3846)    \\\midrule
Total bid liquidity {[}BTC{]} & 2.3766      \\
                              & (2.9920)    \\\midrule
Daily traded volume {[}BTC{]} & 179090.7037 
\end{tabular}
\caption{A brief description of our data set comprising BTCUSDT level 2 LOB data from Binance at December 19, 2022. Whenever there is a quantity in parenthesis, it corresponds to the standard deviation of the element in the corresponding row, where the number above it is the average. We remark that we extracted the ``Daily traded volume'' of a trades data set, distinct from our LOB data.}
\label{tab:data}
\end{table}

\begin{figure}
    \centering
    \includegraphics[scale=0.4]{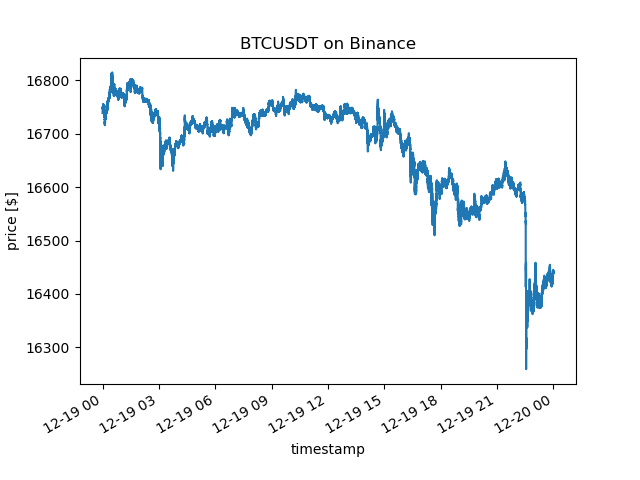}
    \includegraphics[scale=0.4]{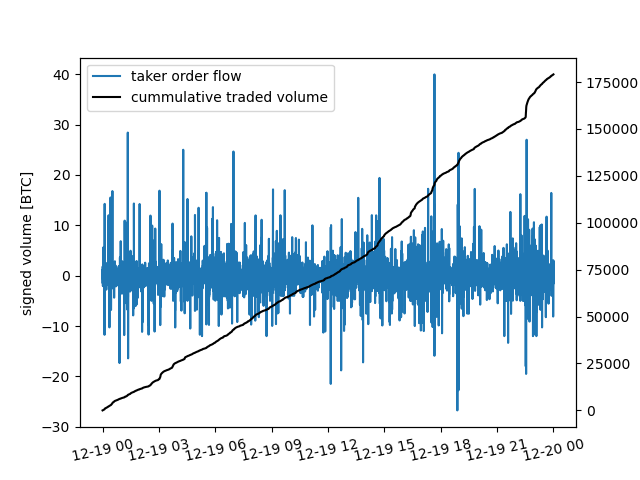}
    \caption{On the left, we present Binance's BTCUSDT price trajectory on December 19, 2022. On the right, we illustrate the taker order flow in this market at this day.}
    \label{fig:btcusdt_path_20231219}
\end{figure}

\subsection{Estimation of the temporary price impact coefficient}

Let us begin by discussing the power-law we stipulate for the temporary price impact. From here on, we will focus on liquidation programs, interacting with the LOB by sending sell market orders; hence, we will concentrating on the bid side of the order book snapshots. At each timestamp $t,$ we take a state $L_t$ of $\mathcal{D}.$ Then, for each hypothetical sell volume $\nu,$ we walk as many layers of $L_t$ as necessary for our trade to be executed, resulting in a realized price per share $R_t(\nu).$ Using Physics jargon, $R_t(\nu)$ is the realized time $t$ price per share corresponding to a \textit{virtual} trade of volume $\nu.$ More rigorously, the unit of $\nu$ is not that of volume, but it represents our theoretically continuous rate of trading (whence it is measured in terms of volume per unit of time), so $R_t(\nu)$ is the result of its instantaneous interaction with the market at time $t.$ In this way, we compute the impact per share $I_t(\nu)$ undergone by this trade as the difference between the current bid price $S_t$ and the realized price per share $R_t(\nu),$ i.e., $I_t(\nu) := S_t - R_t(\nu);$ see Figure \ref{fig:emp_curve} for an illustration. We emphasize that we compute $I_t(\nu),$ for each $\nu,$ directly via the data in $L_t.$ The power-law assumption (see, e.g., \cite[Eq. (8)]{bouchaud2010price} and the references therein) postulates that
\begin{equation} \label{eq:PowerLaw}
    I_t(\nu) = \kappa \nu^\phi.
\end{equation}
We illustrate \eqref{eq:PowerLaw} in Figure \ref{fig:emp_curve_with_fit}. We remark that there is strong evidence for concavity of the temporary price impact in the empirical literature, i.e., $\phi \in \left]0,1\right],$ see \cite{lillo2003master,loeb1983trading}. We will corroborate this stylized fact in our particular experiments.

\begin{figure}
    \centering
    \includegraphics[scale=0.4]{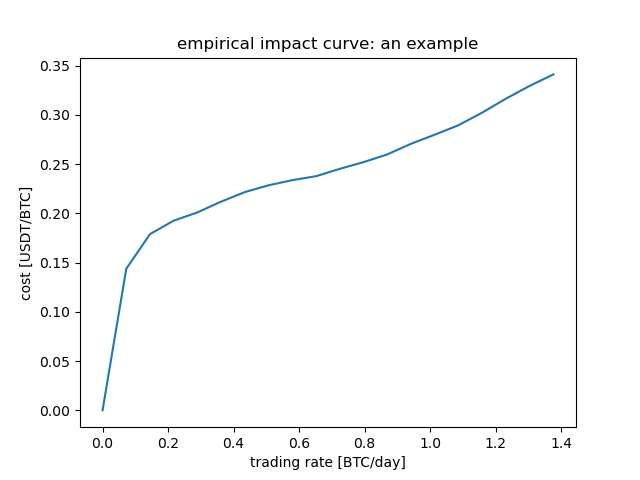}
    \includegraphics[scale=0.4]{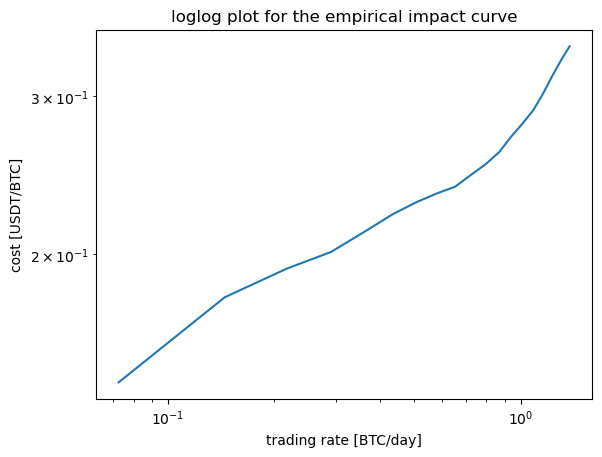}
    \caption{Empirical impact curve $\nu \mapsto I_t(\nu)$ for a given $t.$ On the right, we present the corresponding loglog plot of it.}
    \label{fig:emp_curve}
\end{figure}

\begin{figure}
    \centering
    \includegraphics[scale=0.4]{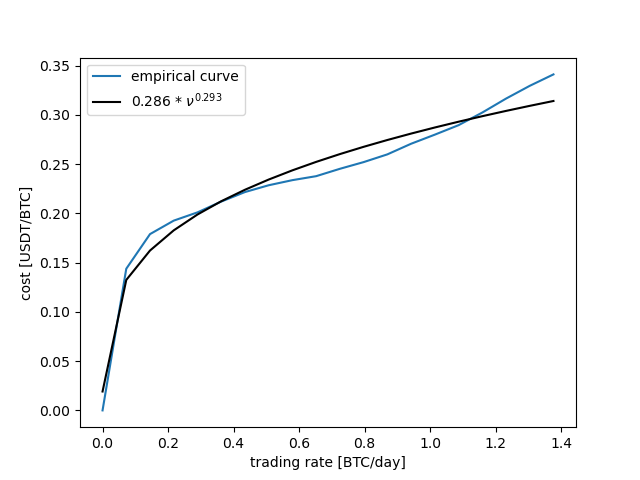}
    \caption{The power law fit for a given impact curve $\nu \mapsto I_t(\nu).$}
    \label{fig:emp_curve_with_fit}
\end{figure}

In this work, we will assume that the exponent $\phi$ in relation \eqref{eq:PowerLaw} is constant, but we will model stochastic liquidity in a similar way as in \cite{almgren2012optimal}: by assuming that $\kappa = \left\{ \kappa_t \right\}_t$ is stochastic. We carry out the estimation of them via a two step procedure, which we now describe.
\begin{itemize}
    \item[\textbf{1.}] We first estimate $\phi$ using a bagging methodology, cf. \cite{breiman1996bagging}. Namely, we fix positive integers $M$ and $N,$ and we create $M$ subsets $\mathcal{D}_1,\,\ldots,\,\mathcal{D}_M$ of $\mathcal{D},$ each of which comprising $N$ order book states sampled randomly but with replacement from $\mathcal{D}$ (bootstrapping). Then, for each $i \in \left\{ 1,\,\ldots,\, M \right\},$ we solve\footnote{We use a BFGS algorithm to solve this minimization problem.}
    $$
    \left(\widehat{\kappa}_i,\,\widehat{\phi}_i\right) = \argmin_{(\kappa,\,\phi)} \sum_{\substack{ L_t \in \mathcal{D}_i,\, \\ \nu \in \left[v_t,V_t\right]}} \left( I_t(\nu) - \kappa |\nu|^{\phi} \right)^2,
    $$
    where we regard $v_t := \inf\left\{ \nu : I_t(\nu) > 0 \right\},$ while we denoted by $V_t$ the sum of all of the bid amounts from the first $L_t$ layer to the last one.\footnote{In particular, since there are $25$ layers, $V_t > v_t$.} Then, we get our estimate $\widehat{\phi}$ of $\phi$ by aggregating:
    $$
    \widehat{\phi} = \frac{1}{M}\sum_{i=1}^M \widehat{\phi}_i. 
    $$
    \item[\textbf{2.}] Next, we fix a lookback period $w \geq 0$. For each timestamp $t,$ we estimate $\kappa_t$ as the slope $\widehat{\kappa}_t$ of the following linear regression\footnote{Here, we use Ordinary Least Squares.}:
    \begin{equation} \label{eq:LinRegForKappa}
        I_{s}(\nu) = \widehat{\kappa}_t |\nu|^{\widehat{\phi}} + \eta_{s} \hspace{1.0cm} \left(s \in \left[(t-w)_+,t\right] \text{ such that } L_s \in \mathcal{D},\, \nu \in \left[v_s,\,V_s\right]\right).
    \end{equation}
    Above, for each timestamp $s,$ we consider $v_s$ and $V_s$ as in the previous step, and we have written $(t-w)_+ := \max\left\{t-w,0\right\}$.
\end{itemize}

The bagging methodology we conduct in step one above seems adequate because it fits exponents for various batch of books, whence we expect the exponent $\widehat{\phi}$ we estimated to work decently in a uniform manner. From Table \ref{tab:exponent_estimation_params}, we also see that the variance in our estimate is quite small, indicating an adequate fit. Bagging methods are commonly appropriate to reduce predictors' variance and reduce overfitting; we again refer to \cite{breiman1996bagging} and the references therein for a more detailed account.\footnote{It is also worth mentioning that it points out how bagging works well for unstable procedures. It seems to be the case for financial high-frequency settings, where we have the presence of microstructural noise.} 

% \usepackage{booktabs}
\begin{table}[]
\centering
\begin{tabular}{@{}cccccc@{}}
\toprule
$\widehat{\phi}$ & 1\%    & 25\%   & median & 75\%   & 99\%   \\ \midrule
0.2833           & 0.2631 & 0.2763 & 0.2847 & 0.2911 & 0.3046 \\
(0.0116)         &        &        &        &        &        \\ \bottomrule
\end{tabular}
\caption{Our bagging estimate $\widehat{\phi}$ of the exponent $\phi$ and some of its corresponding quantiles. In parenthesis, the standard deviation when aggregating to form the final estimate. Here, we fixed $M=20,$ $N = 2000,$ and $w=1$ second.}
\label{tab:exponent_estimation_params}
\end{table}

Regarding step two, taking $w=0$ leads to performance of the regressions in \eqref{eq:LinRegForKappa} in an update-by-update manner. Using a time window $w>0$ is in line with what we did in the previous step. It will aggregate a few order books for each update time $t,$ and provide an estimate working for all of them throughout a certain (small) time frame. Not only using $w>0$ helps filter out microstructural noise, but it is also consistent with the fact that the trader is subject to latency, so whenever she wants to interact with the LOB, say at time $t,$ she will do so with an uncertain state $L_s,$ with $s \in \left[t,\,t+w\right],$ for some $w$ (the latency\footnote{We remark that latency is stochastic itself. Here, we can think of a constant $w$ as its average value, for instance.} of her infrastructure). 

In Figure \ref{fig:tic_paths}, we present an example of the estimated path for $\left\{\widehat{\kappa}_t\right\}_t$ for our reference data set. In all examples of this work, we fix $M=20$\footnote{We also ran the estimations with more trials (greater $M$) but it did not yield an estimate too far from the current one.} and $N=2000$ in step one of our bagging estimation algorithm, as well as a lookback window of $w=1$ second. 

\begin{figure}
    \centering
    \includegraphics[scale=0.4]{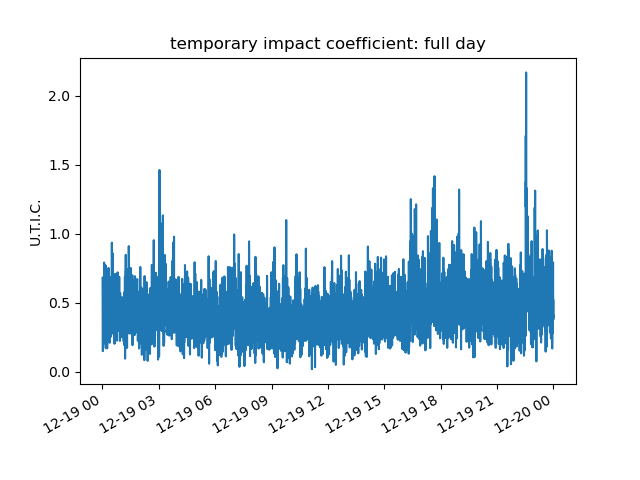}
    \includegraphics[scale=0.4]{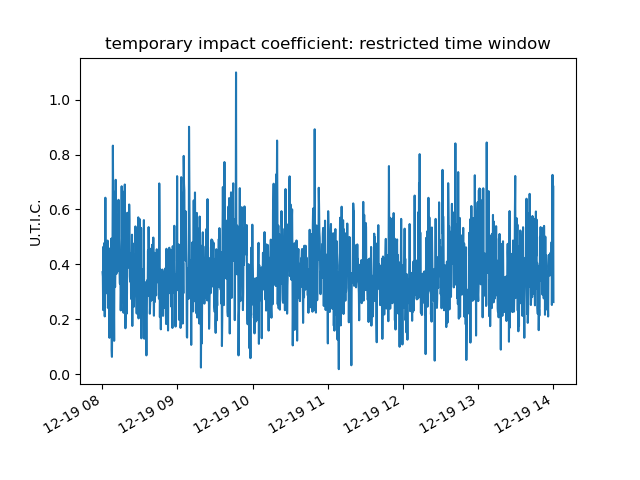}
    \caption{Realization of the process $\kappa$ over December 19, 2022 for the market BTCUSDT traded on the exchange Binance. The plot we refer to as ``restricted'' comprehends only the time window between $8$ AM and $2$ PM, where that mean value of the temporary price impact seems to be more stable. In the plots above, we have written ``U.T.I.C.'' meaning ``Units of Temporary Price Impact Coefficient'', given by: $[ \kappa ] = [\text{cash}]\times [\text{time}]^\phi \times [\text{volume}]^{-(1+\phi)}$. We derive the U.T.I.C. from \eqref{eq:PowerLaw}, using that $[\nu]$ is units of volume per units of time, whereas $[I_t(\nu)]$ is units of cash per units of volume. In the present illustrations, our unit of cash is one USDT, our unit of time is one day, and our unit of volume is one BTC.
    }
    \label{fig:tic_paths}
\end{figure}

\subsection{Estimation of intraday volatility}

We estimate the intraday volatility by applying the Two-Scale Realized Variance (TSRV) method of Zhang, Mykland, and Aït-Sahalia \cite{zhang2005tale}. We remark it is also a bagging-like estimator, which is based on averaging the estimated variance on subsamples, and then correcting the biases. The work \cite{gatheral2010zero} studies several such volatility estimators --- among which TSRV. From their results, we expect this procedure to yield a decent estimation in the face of microstructure noise. Since we focus on liquidation programs, we estimate the asset's bid price volatility. We fix a lookback time $\Delta,$ and for each update time $t,$ we gather all bid prices over $\left[(t-\Delta)_+,\,t\right]$ whenever they change\footnote{We do not sample repeated prices since in this way we typically obtain better estimates.} and use them to get the TSRV estimate $\Sigma_t^2.$ In the TSRV estimator, we employ a maximum subsampling bandwidth of size five. In order to assess the reasonableness of our estimate, we form the price differences
$$
Z_t = \frac{p_{t} - p_{(t-\Delta)_+}}{\Sigma_t},
$$
where $\left\{p_t\right\}_t$ is the bid price time series. Henceforth, we fix $\Delta = 1$ minute. We present statistics of the random variables consisting of samples of the process $\left\{ Z_t \right\}_t$ in Table \ref{tab:stats_samples}, and some corresponding illustrative plots in Figures \ref{fig:hists_normal}, \ref{fig:qq_plots_normal} and \ref{fig:sqrt_tsrv}, both for the whole day and a restricted six hour time window from $8$ AM to $2$ PM. During this time window, the mean of volatility tends to be more stable,\footnote{We can make the overall mean more stable, e.g., by de-seasonalizing the volatility. The work \cite{fouque2021optimal} discusses extracting a seasonal profile, which can be replicated here. For the $8$ AM to $2$ PM time frame, assuming that the mean of the volatility is stable is quite reasonable for Binance's BTCUSDT market, at least as of the period consisting of December 2022 days. We could also model the mean of the volatility itself as being stochastic; the techniques we develop here also serve to treat this case. We choose not to do pursue such endeavors for the sake of simplicity.} so from here on we will focus on it. From those tables and figures, we see that the estimate constrained to the restricted time window is also rather decent. In view of Table \ref{tab:stats_samples}, the empirical variances of $\left\{ Z_t \right\}_t$ are reasonably close to one, and one can check that this is consistent at least throughout the days of December 2022.

% \usepackage{booktabs}
\begin{table}[]
\centering
\begin{tabular}{@{}ccc@{}}
\toprule
Time window & Sample mean & Sample variance \\ \midrule
Full day    & 0.0142      & 0.8865          \\
Restricted  & 0.0139      & 0.7888          \\ \bottomrule
\end{tabular}
\caption{Here, we present a few statistics for the random variables $Z$ sampled once each $15$ seconds. The restricted time window comprehends trading from $8$ AM to $2$ PM.  }
\label{tab:stats_samples}
\end{table}

\begin{figure}
    \centering
    \includegraphics[scale=0.4]{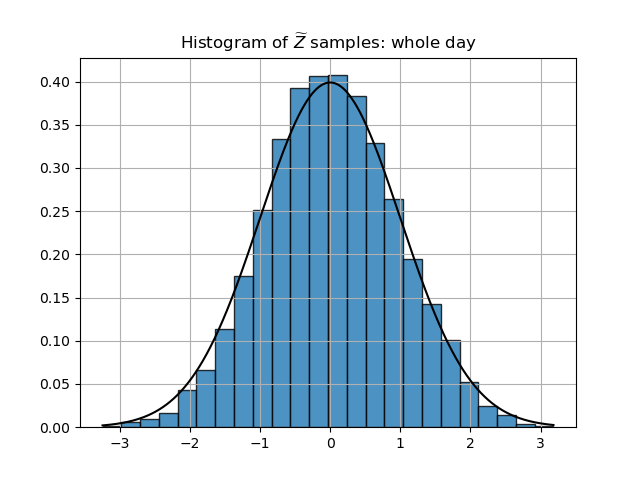}
    \includegraphics[scale=0.4]{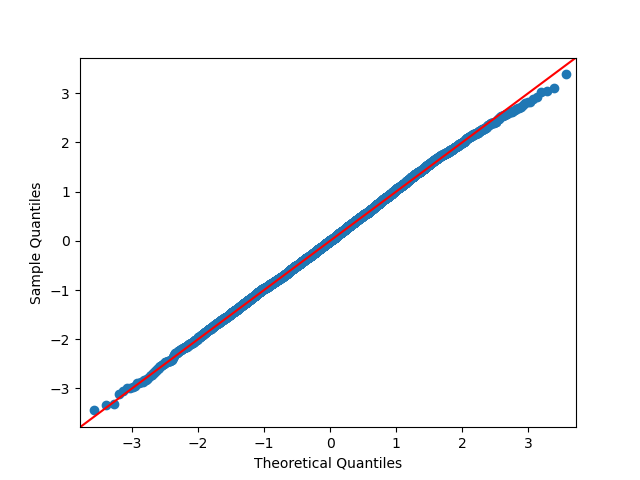}
    \caption{Histogram and QQ-plot for the normalized $Z$ random variables $\widetilde{Z}=Z / \mathbb{V}(Z)$ over the full day. Here, we compute the variance $\mathbb{V}(Z)$ of $Z$ by sampling $Z$ over the full day.}
    \label{fig:hists_normal}
\end{figure}

\begin{figure}
    \centering
    \includegraphics[scale=0.4]{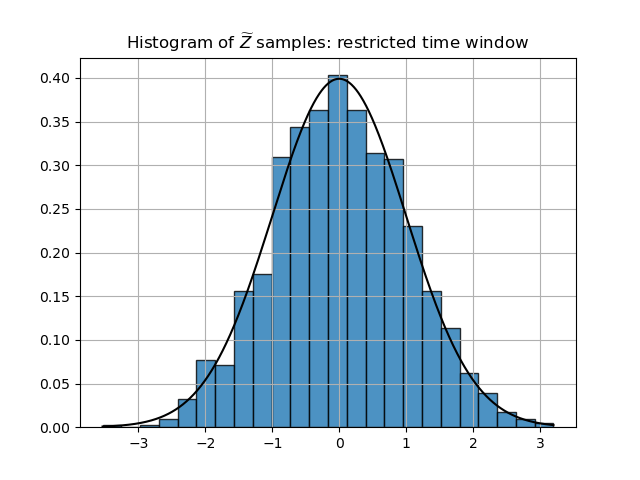}
    \includegraphics[scale=0.4]{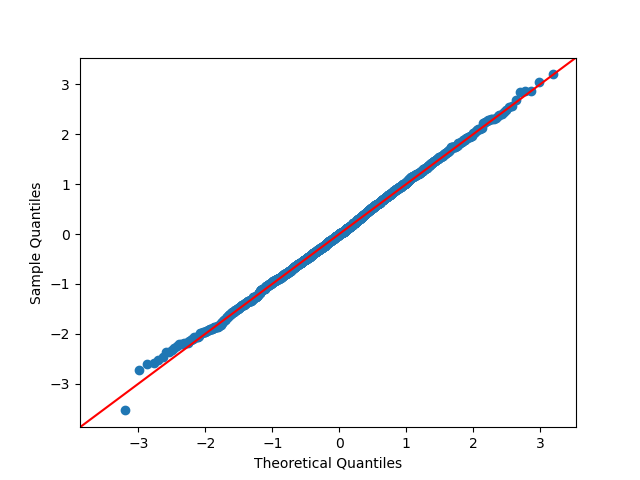}
    \caption{Histogram and QQ-plot for the normalized $Z$ random variables $\widetilde{Z}=Z / \mathbb{V}(Z)$ over the restricted time window.  Here, we compute the variance $\mathbb{V}(Z)$ of $Z$ by sampling $Z$ over the restricted time window.}
    \label{fig:qq_plots_normal}
\end{figure}

\begin{figure}
    \centering
    \includegraphics[scale=0.4]{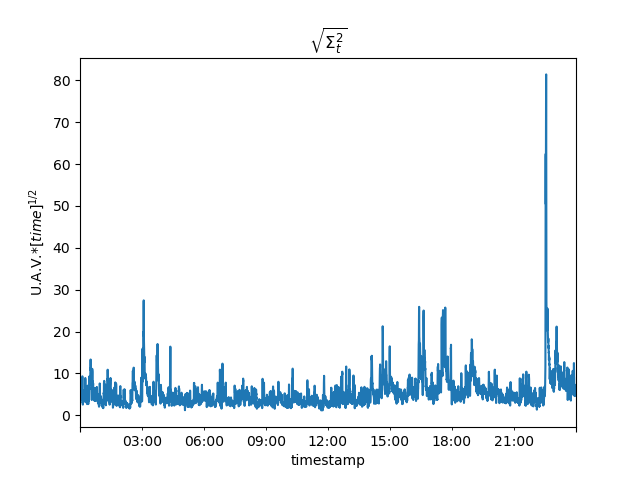}
    \includegraphics[scale=0.4]{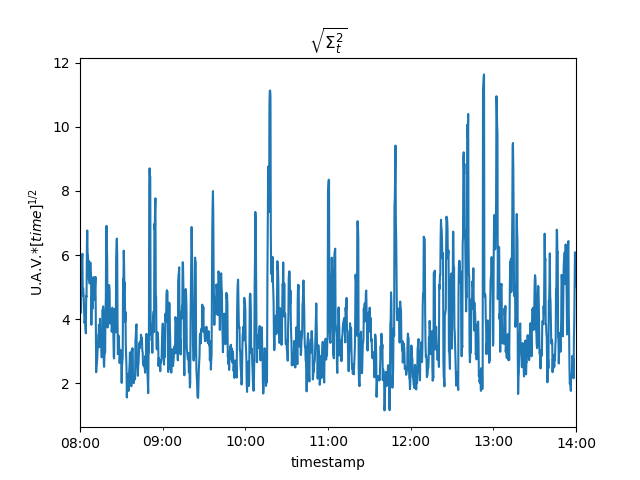}
    \caption{Realization of the square root of the TSRV's realized variance on December 19, 2022 for the bid price. Here, we used ``U.A.V.'' to represent ``Units of Arithmetic Volatility'', which are $[\sigma] = [\text{cash}] \times [\text{volume}]^{-1} \times \text{time}^{-1/2}.$ We can derive this from the fact that the unit of the price is units of cash per unit of volume, whereas the volatility $\sigma$ is related to the integrated variance $\Sigma^2$ as $\sigma=\sqrt{\Sigma^2/\Delta}.$ Presently, our unit of cash is one USDT, our unit of time is one day, and our unit of volume is one BTC.}
    \label{fig:sqrt_tsrv}
\end{figure}

Our natural estimate $\widehat{\sigma}_t$ for the intraday volatility at time $t$ is thus
$$
\widehat{\sigma}_t = \sqrt{\frac{\omega^2\Sigma_t^2}{\Delta}},
$$
where $\omega^2$ is the empirical variance of $\left\{ X_t \right\}_t$ over the restricted time window from $8$ AM to $2$ PM.\footnote{We make the in-sample correction the TSRV by $\omega^2$ so as to make the subsequent paremeter estimations better. In practice, we do not need to make it, as $\omega^2$ is usually sufficiently close to one (as it is now), whence the original estimate uses to be quite decent.} We present in Figure \ref{fig:logvol_fullday_restricted} plots of our estimation of the log-volatility process for the full trading day and for the restricted time frame.

\begin{figure}
    \centering
    \includegraphics[scale=0.4]{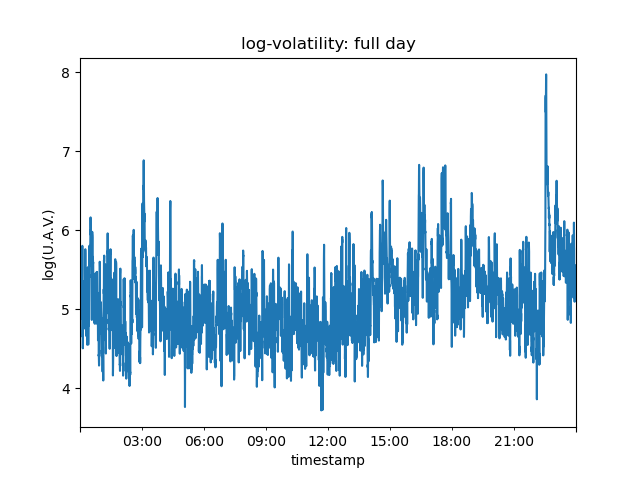}
    \includegraphics[scale=0.4]{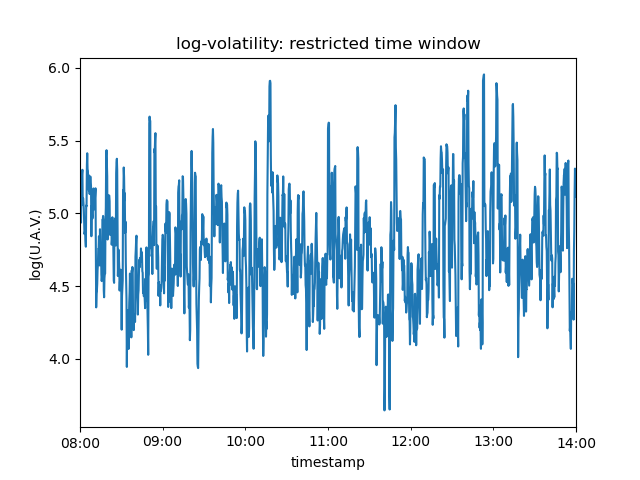}
    \caption{Realized log-volatility process for both the full day and the restricted time window. We employ $\Delta = 1$ minute. Here, our unit of cash is one USDT, and our unit of time is one day.}
    \label{fig:logvol_fullday_restricted}
\end{figure}

Relating this estimate of $\widehat{\sigma}$ to that of $\widehat{\kappa}$ in the previous subsection, we present empirical correlation $\rho$ between those process is
\begin{equation} \label{eq:emp_corr_kappa_logvol}
    \rho = 20.97\%.
\end{equation}
In particular, we notice that it is positive, which makes sense: higher (resp., lower) volatility and lower (resp., higher) liquidity\footnote{Lower (resp., higher) liquidity corresponds to higher (resp., lower) trading costs, meaning higher (resp., lower) values of $\kappa$.} are commonly associated in cryptocurrencies (in general, in highly volatile markets). 

\subsection{Fast mean-reversion}

We now argue that it is reasonable to expect to model both the temporary price impact coefficient and the uncertain intraday volatility as fast mean-reverting. We do so by fitting Ornstein-Uhlenbeck (OU) processes to the corresponding estimated data, wherefrom we will see that their speeds of mean reversion are sufficiently high. In general, given a OU process $x = \left\{x_t \right\}_{t \in \left[0,T\right]},$ $T>0,$ with speed of mean reversion $\lambda_x > 0,$ long-run mean $m_x \in \mathbb{R}$ and diffusion coefficient\footnote{We avoid to call $\eta_x$ ``volatility'' here so as not to confuse it with the price volatility $\sigma$ we were discussing before, which is central to the current work.} $\eta_x > 0,$ i.e.,
$$
dx_t = \lambda_x(m_x-x_t)\,dt + \eta_x\,dW_t,
$$
for a given Brownian motion $\left\{W_t\right\}_t$, we estimate these parameters as in \cite[Eq. (49)]{holy2018estimation}. Namely, we run the ARMA(1,1) regression
\begin{equation} \label{eq:RegressionForOU}
    x_{t_i} = a + b x_{t_{i-1}} + c \epsilon_{i-1} + \epsilon_i,
\end{equation}
for a time sampling $0 = t_0 < \ldots < t_N = T,$ with $t_i - t_{i-1} = dt$ ($dt$ being independent of $i$), and some i.i.d. $N(0,\gamma^2)$ random variables $\left\{\epsilon_i\right\}.$ We then set
\begin{equation} \label{eq:EstimatedParamsOU}
    \widehat{\lambda}_x = -\frac{\log(b)}{dt},\, \widehat{m}_x = \frac{a}{1-b} \text{ and } \widehat{\eta}_x = \gamma\sqrt{-\frac{2\left(b + b c^2 + b^2 c  + c \right)\log(b)}{dt (1-b^2) b}}.
\end{equation}
We apply this technique to $\widehat{\kappa} = \left\{\widehat{\kappa}_t\right\}_t$ and $\log\left(\widehat{\sigma}\right) = \left\{\log\left(\widehat{\sigma}_t\right)\right\}_t,$ see Figure \ref{fig:panels_kappa_logvol}. In carrying out the estimates, we restrict the processes to the window starting at $8$ AM and ending at $2$ PM, where we can see from Figures \ref{fig:tic_paths} and \ref{fig:logvol_fullday_restricted} that their long-term means are more stable.\footnote{We could also seek modelling the processes using double OU processes, where means are themselves mean-reverting (possibly slowly). We can approach the problem under this assumption with the same techniques we use here, cf. \cite{fouque2011multiscale}.} Moreover, we sample the processes once each $15$ seconds --- which leads to more stable estimates --- filtering out some microstructure noise. Hence, under those constraints, we model $\kappa$ as an OU process, whereas regarding $\sigma$ as an expOU one. More precisely, we run the regression \eqref{eq:RegressionForOU} using $\left\{ x_t\right\}_t$ as either the process $\kappa$ itself, or the natural logarithm of $\sigma$, in both cases with a proper downsampling, and then we showcase in Table \ref{tab:fast_mr_params} the parameters we estimate for them according to \eqref{eq:EstimatedParamsOU}. Our results corroborate the claim that we can safely regard them both as fast mean-reverting. 

\begin{figure}
    \centering
    \includegraphics[scale=0.5]{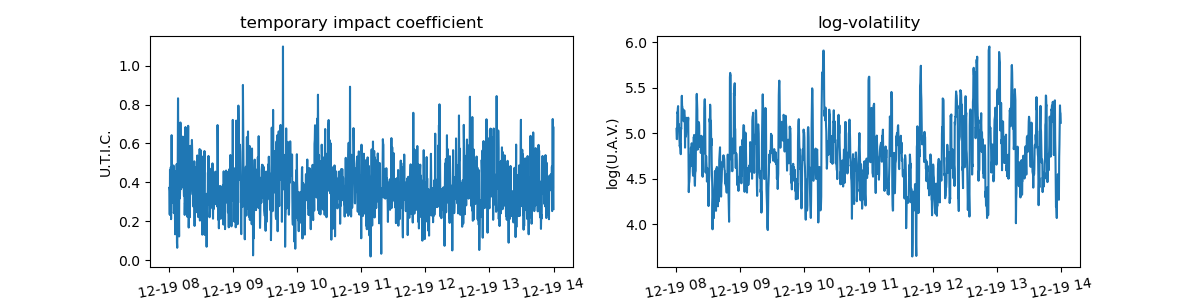}
    \caption{Panels comprising the estimated realizations of $\kappa$ and $\log(\sigma)$ on December 19 for the BTCUSDT market on Binance. In the above panels, our cash unit is one USDT, our time unit is one day, and our volume unit is one BTC.}
    \label{fig:panels_kappa_logvol}
\end{figure}

% \usepackage{booktabs}
\begin{table}[]
\centering
\begin{tabular}{@{}cccc@{}}
\toprule
$x$                & $\lambda_x$ & $m_x$   & $\eta_x$ \\ \midrule
$\kappa$           & $1905.2180$ & $0.3782$ & $4.0134$   \\
$\log(\sigma)$     & $1279.7954$ & $4.7810$ & $19.0326$  \\ \bottomrule
\end{tabular}
\caption{High-frequency estimation of mean-reversion parameters. In doing so, we sample each process $x$ once each $15$ seconds. Our time unit is days, whence $dt = 15/(60*60*24).$ }
\label{tab:fast_mr_params}
\end{table}

\subsection{Related literature}

The paper \cite{almgren2012optimal} deals with a model comprising stochastic temporary price impact, but assuming linearity, i.e., that $\phi = 1$ in \eqref{eq:PowerLaw}. They also allow volatility to be stochastic and devise a numerical method for computing the optimal strategy under suitable assumptions. In a discrete-time setting, the work \cite{cheridito2014optimal} models stochastic volatility and liquidity as independent processes in a Markov chain. See also \cite{walia2006optimal} for a discrete-time discrete-space solution for the problem under discussion. In a game-theoretic framework, \cite{evangelista2020finite} considers a market model with stochastic volatility and liquidity. We also refer to \cite{souza2021regularized} for theoretical and numerical results about the optimal strategy of the model we will investigate in the current work, but not necessarily in an ergodic setting. 

Some other works consider stochastic price impact only, regarding volatility as being constant. The use of jump processes for modeling the stochastic price impact is the approach of \cite{moazeni2013optimal,bayraktar2011optimal} --- see also other frameworks for studying stochastic price impacts in \cite{becherer2018optimal,fruth2019optimal}. The work in \cite{konishi2002optimal} addresses optimal slicing of VWAP orders under stochastic volatility without considering price impact. The papers \cite{graewe2018smooth,horst2020continuous} allow for the uncertainty of both the price impact and the risk aversion --- the latter including stochastic volatility (if we assume that the urgency parameter of the trader is proportional to the variance of the asset price, say). A few other efforts model uncertain resilience, such as \cite{siu2019optimal}, extending the OW model, under regime-switching stochastic resilience, and also \cite{graewe2017optimal}. As for applying fast mean-reversion asymptotic techniques to problems in finance, we mention the standard monograph \cite{fouque2011multiscale} and the references therein. The work \cite{fouque2021optimal} applies such techniques to optimal trading, but assuming that volatility is constant and a linear impact ($\phi = 1$) in \eqref{eq:PowerLaw}.  

\subsection{Our contributions}

We consider a model with stochastic liquidity, modeling it as an uncertain temporary price impact subject to the power-law \eqref{eq:PowerLaw}, determining the randomly varying coefficient. Together with the latter parameters, we also allow volatility to be stochastic, and we assume a multi-dimensional Markov diffusion drives their dynamics. We concentrate on the class of assets for which it is realistic to regard the speeds of mean-reversion towards a long-run level as being sufficiently large, in a way to be made precise. The reference framework under which we carry out our numerical experiments is when this Markov diffusion is a two-dimensional OU process.

The way we identify the optimal trading strategy is the same as in \cite{souza2021regularized}. Fortunately, the rate we obtain in the regularized problem is uniformly bounded with respect to the small parameter with respect to which we wish to develop our asymptotic analysis. We begin our investigation by conducting a formal asymptotic analysis, from where we will derive a leading-order approximation for the optimal trading strategy. We proceed to provide some numerical illustrations --- using the parameters we obtained from our estimations --- to illustrate the behavior of the trading rate we derived. Then, we continue our formal analysis to derive the first-order correction to our approximately optimal strategy. We then carry out numerical assessments, analogous to the ones we previously discussed, but we construct now for the first-order approximation.

Finally, we provide some accuracy results establishing that the two approximations we obtained do have the order of approximation we expect of them. The idea of both proofs is to linearize the equations in a way to make feasible the application of the usual Feynman-Kac Theorem. We obtain ``reflexive'' representations for the error terms, i.e., representations of these as fixed-point relations. Under the suppositions we make, we are apt to carry out suitable estimates and employ Gronwall's Lemma to deduce the asymptotics we desire. For the leading-order approximation, we prove a pointwise result in a somewhat direct manner. For the first-order correction, we provide a result on the size of the error term computed over the paths of the multi-dimensional driver. In order for the latter to hold uniformly with respect to time, we need the aid of an appropriate weight. From this, a pointwise result uniformly away from the terminal time follows. We are also apt to show that the desired accuracy for the first-order correction holds on the homogeneous average in time as well.

\subsection{Structure of the paper}

We organize the remainder of the paper as follows. We finish this introductory Section by fixing some notations and terminologies. Then, we present the details of our model and describe some results established elsewhere, with appropriate references, in Section \ref{sec:theModel}. In Section \ref{sec:OrderZero}, we carry out the formal analysis for the derivation of the leading-order approximation, as well as corresponding numerical experiments. We do a similar procedure regarding the first-order correction in Section \ref{sec:OrderOne}. In Section \ref{sec:accuracy}, we give accuracy results for the approximations we derived. We provide our conclusions in Section \ref{sec:conclusions}.

\subsection{Some notations and terminologies}

\begin{itemize}
    \item Henceforth, we fix the terminal time horizon $T>0,$ as well as a complete filtered probability space $\left(\Omega,\mathcal{F},\mathbb{F} = \left\{ \mathcal{F}_t \right\}_{0\leqslant t \leqslant T},\mathbb{P}\right),$ with $\mathcal{F}_T = \mathcal{F}.$ We suppose that this space supports a one-dimensional Brownian motion $B,$ and also a $d-$dimensional one $\boldsymbol{W} = \left(W^1,\ldots,W^d \right)^\intercal,$ where $d \geqslant 1.$ We consider $\mathbb{P}$ as the statistical (or historical measure) --- we will work under it throughout the present work, writing all the expectations (including the conditional ones) under $\mathbb{P}.$ Moreover, for a multi-dimensional Markovian process $\boldsymbol{x} = \left\{\boldsymbol{x}_u\right\}_u,$ we put
    $$
    \mathbb{E}_{t,\boldsymbol{x}}\left[ \cdot \right] := \mathbb{E}\left[ \cdot | \boldsymbol{x}_t = \boldsymbol{x} \right] \hspace{1.0cm} (t \in \left[0,T\right]).
    $$
    \item For a probability measure $\Pi$ on the Euclidean space $\mathbb{R}^d,$ let us write $f \in L^1\left( \Pi \right)$ if, and only if, $f$ is measurable and $\int_{\boldsymbol{y} \in \mathbb{R}^d } |f(\boldsymbol{y})| \Pi(d\boldsymbol{y}) < \infty.$ In this case, we write $\left\langle f \right\rangle := \int_{\boldsymbol{y}\in \mathbb{R}^d} f(\boldsymbol{y})\Pi(d\boldsymbol{y}).$
    \item The letter $C$ denotes a generic positive constant, which may change from line to line within estimates. Unless we state otherwise, $C$ possibly depends on all model parameters.
    \item We write $f = g + O(h),$ for three functions $f,\,g$ and $h,$ to mean that $|f-g| \leqslant C |h|$ pointwise. Whenever $h$ is a model parameter (thus a constant function), we allow $C$ to depend on the point which we calculate $f-g.$ Generally, in case we want to emphasize the dependence of $C$ on a variable $\theta,$ we write $f = g + O_\theta(h).$
    \item We will consider, for each $t \in \left[0,T\right],$ the admissible control set $\mathcal{U}_t$ comprising the $\mathbb{F}-$progressively measurable processes $\nu = \left\{ \nu_u \right\}_{t\leqslant u \leqslant T}$ such that $\mathbb{E}\left[ \int_t^T \left(\nu_u\right)^2\,du\right] < \infty.$ 
\end{itemize}

\section{The model} \label{sec:theModel}

\subsection{Dynamics of the state variables}

Beginning at a time $t \in \left[0,T\right],$ we consider an agent who is negotiating a financial instrument with price process\footnote{As long as the resulting strategy does not lead to price manipulation, as in \cite[Corollary 3.10]{souza2021regularized}, we can consider $S$ as the bid (respectively, ask) price for a liquidation (respectively, acquisition) execution program, as we did in Section \ref{sec:intro}.} $S = \left\{ S_u \right\}_{t\leqslant u \leqslant T}$ satisfying
$$
\begin{cases}
dS_u = \sigma_u\,dB_u,\\
S_t = s .
\end{cases}
$$ 
for a volatility process $ \left\{ \sigma_u \right\}_u.$ We denote the trader's turnover rate at time $u \in \left[t,T\right]$ by $\nu_u,$ whence her inventory holdings evolve according to
$$
\begin{cases}
dQ^\nu_u = \nu_u\,du,\\
Q^\nu_t = q,
\end{cases}
$$
where we assume that her initial inventory $q$ is given. The agent incurs a temporary price impact whose value per share is proportional to $\left|\nu \right|^\phi,$ for some $\phi \in \left]0,1\right],$ in such a way that her execution price per share $\widehat{S}^\nu$ at time $u$ is
$$
\widehat{S}^\nu_u = S_u + \kappa_u\left|\nu_u\right|^\phi\sgn\left(\nu_u\right).
$$
We emphasize that we allow $\left\{ \kappa_u \right\}_{t\leqslant u \leqslant T}$ to be a stochastic process above. The resulting agent's cash process is thus
$$
\begin{cases}
dX^\nu_u = -\widehat{S}^\nu_u \nu_u\,dt = -S_u \nu_u \,du - \kappa_u|\nu_u|^{1+\phi}\,du,\\
X^\nu_t = x.
\end{cases}
$$

From now on we assume (with slight abuse of notations) that $\kappa$ and $\sigma$ are such that
$$
\kappa_u = \kappa\left( \boldsymbol{y}_u\right) \text{ and } \sigma_u = \sigma\left( \boldsymbol{y}_u \right),
$$
for suitable deterministic continuous functions $\kappa, \sigma : \mathbb{R}^d \rightarrow \mathbb{R},$ and a $d-$dimensional Markov diffusion $\boldsymbol{y}:$
\begin{equation} \label{eq:EvolutionOfFactorEps}
    \begin{cases}
        d\boldsymbol{y}_u = \frac{1}{\epsilon}\boldsymbol{\alpha}\left(\boldsymbol{y}_u\right)\,du + \frac{1}{\sqrt{\epsilon}}\boldsymbol{\beta}\left(\boldsymbol{y}_u\right)\,d\boldsymbol{W}_u,\\
        \boldsymbol{y}_t = \boldsymbol{y},
    \end{cases}
\end{equation}
where $\boldsymbol{W} = \left(W^1,...,W^m\right)^\intercal$ is an $m-$dimensional Brownian motion, whereas $\boldsymbol{\alpha} : \mathbb{R}^d \rightarrow \mathbb{R}^d$ and $\boldsymbol{\beta} : \mathbb{R}^d \rightarrow \mathbb{R}^{d\times m}$ are two deterministic functions. In case we want to emphasize the dependence on $\epsilon$ in \eqref{eq:EvolutionOfFactorEps}, we write $\left\{ \boldsymbol{y}_u \right\} \equiv \boldsymbol{y}^\epsilon \equiv \left\{ \boldsymbol{y}^\epsilon_u \right\}_{t\leqslant u \leqslant T}.$ Let us observe that 
$$
\boldsymbol{y}^\epsilon \stackrel{d}{=} \left\{ \boldsymbol{y}^1_{u/\epsilon} \right\}_{ 0\leqslant u \leqslant T }.
$$
However, in any circumstance where we refer to the process $\left\{ \boldsymbol{y}_u \right\},$ with no superscript, we mean $\left\{ \boldsymbol{y}^\epsilon_u \right\}_{t \leqslant u \leqslant T},$ where $t$ shall be clear from the context.

The trader's wealth $w^\nu_u$ at time $u$ consists of her current cash holdings $X^\nu_u$ plus the book value of her current inventory $Q^\nu_uS_u,$ i.e., $w^\nu_u := X^\nu_u + Q^\nu_uS_u.$ Thus, it is straightforward to derive that
$$
w^\nu_T = w^\nu_t - \int_t^T \kappa_u |\nu_u|^{1+\phi} \,du + \int_t^T \sigma_u Q^\nu_u\,dB_u.
$$

Henceforth, we rely on the following assumptions:
\begin{itemize}
    \item[(\textbf{H1})] The functions $\boldsymbol{\alpha}$ and $\boldsymbol{\beta}$ are Lipschitz continuous.
    \item[(\textbf{H2})] Both $\kappa$ and $\sigma$ are continuous functions and there are $\underline{\kappa},\overline{\kappa},\overline{\sigma}>0,$ $\underline{\sigma} \geqslant 0,$ such that $\overline{\kappa} \geqslant \kappa \geqslant \underline{\kappa}$ and $\overline{\sigma} \geqslant \sigma \geqslant \underline{\sigma}.$ Moreover, the exponent $\phi$ of the power-law assumption belongs to $\left]0,1\right],$ and the parameter $\epsilon$ is positive.
\end{itemize}

\subsection{Performance criteria and the value function}

The performance criteria of the trader consist of the difference between her terminal and initial wealth, along with some penalizations for holding inventory. More precisely,
\begin{align*}
    J^\nu\left( t,\, q,\, \boldsymbol{y} \right) &:= \mathbb{E}_{t,s,q,x,\boldsymbol{y}}\left[ X^\nu_T + Q^\nu_T S^\nu_T - \left(x+qs\right) - \gamma \int_t^T \sigma_u^{1+\phi}\left|Q^\nu_u\right|^{1+\phi}\,du - A\left|Q^\nu_T\right|^{1+\phi} \right] \\
    &= -\mathbb{E}_{t,q,\boldsymbol{y}}\left[ \int_t^T \left\{ \kappa_u\left| \nu_u \right|^{1+\phi} + \gamma \sigma_u^{1+\phi}\left|Q_u^\nu\right|^{1+\phi} \right\}\,du + A \left|Q_u^\nu\right|^{1+\phi} \right].
\end{align*}
The corresponding value function $J,$
\begin{equation} \label{eq:OptControlProb}
    J := \sup_{\nu \in \mathcal{U}_t} J^\nu,
\end{equation}
is a viscosity solution of the HJB
$$
\partial_t J + \frac{1}{\epsilon}\mathcal{L}J + \kappa H\left( \partial_q J / \kappa \right) - \gamma \sigma^{1+\phi} |q|^{1+\phi}= 0,
$$
with terminal condition $J|_{t=T} = - A |q|^{1+\phi},$ see \cite{souza2021regularized}, where the operator $\mathcal{L}$ is the infinitesimal generator of $\boldsymbol{y}$ when $\epsilon = 1,$ i.e.,
$$
\mathcal{L} = \frac{1}{2}\tr\left( \boldsymbol{\beta}\left(\boldsymbol{y}\right) \boldsymbol{\beta}\left(\boldsymbol{y}\right)^\intercal D_{\boldsymbol{y}}^2 \right) + \boldsymbol{\alpha}\left(\boldsymbol{y}\right) \cdot D_{\boldsymbol{y}}, 
$$
and $H(p) := \phi \left[ |p|/(1+\phi) \right]^{1+1/\phi}.$ We define the domain $\mathcal{D}\left(\mathcal{L} \right)$ of $\mathcal{L}$ as\footnote{We write $C\left(\mathbb{R}^d\right)$ to denote the space of functions $g : \mathbb{R}^d \rightarrow \mathbb{R}$ which are continuous.} 
$$
\mathcal{D}\left(\mathcal{L} \right) := \left\{ g \in C\left( \mathbb{R}^d \right) : \text{ the limit } \lim_{t \downarrow 0}\frac{1}{t}\left( \mathbb{E}\left[ g(\boldsymbol{y}^1_t) | \boldsymbol{y}^1_0 = \boldsymbol{y}\right] - g(\boldsymbol{y}) \right) \text{ exists uniformly in } \boldsymbol{y} \in \mathbb{R}^d \right\}.
$$ 
We envisage proceeding in a suitable ergodic framework --- specifically, the one we find described in \cite[Subsection 3.2]{fouque2011multiscale}. Thus, we fix the further hypotheses:
\begin{itemize}
    \item[(\textbf{H3})] The operator $\mathcal{L}$ has a discrete spectrum with a positive gap, i.e., zero is an isolated eigenvalue. We also suppose that the remaining eigenvalues $\left\{a_k\right\}$ of $\mathcal{L}$ satisfy $0 > a_1 > a_2 > \ldots,$ and that they are all simple.
    \item[(\textbf{H4})] The process $\boldsymbol{y}^1$ with infinitesimal generator $\mathcal{L}$ has a unique invariant distribution $\Pi.$ Moreover, we assume that $\boldsymbol{y}^1$ has moments of all orders, bounded uniformly in time. 
    \item[(\textbf{H5})] For each continuous at most polynomially growing $f \in L^1(\Pi)$ that is centered, i.e., $\left\langle f \right\rangle = 0,$ the Poisson equation $\mathcal{L}v = f$ admits at most polynomially growing solutions $v \in \mathcal{D}\left( \mathcal{L} \right) \cap L^1(\Pi).$
\end{itemize}

\begin{remark} \label{rem:expression_eigen}
The normalized eigenfunctions $\left\{ \psi_k \right\}$ of $\mathcal{L}$ are those that satisfy $\mathcal{L}\psi_k = a_k \psi_k$ and $\left\langle \psi_k^2 \right\rangle = 1.$ They form a basis of $L^2(\Pi),$ and for our operator $\mathcal{L}$ we have $\psi_0 \equiv 1.$ Moreover, we can express every $g \in L^2(\Pi)$ as
$$
g(\boldsymbol{y}) = \sum_{k} c_k\psi_k(\boldsymbol{y}),
$$
where $c_k = \left\langle g\psi_k \right\rangle,$ see \cite[Eq. (3.10)]{fouque2011multiscale}.
\end{remark}
\begin{remark} \label{rem:expectation_centered_bdd_eps}
Regarding (\textbf{H3}), we notice as in \cite[Eq. (3.11)]{fouque2011multiscale} that, whenever $g \in \mathcal{D}\left(\mathcal{L}\right)$ with $\left\langle g \right\rangle = 0,$ 
\begin{equation} \label{eq:AsymptoticAvg}
    \mathbb{E}_{t,\boldsymbol{y}}\left[ g\left(\boldsymbol{y}^\epsilon_T \right)\right] = \left\langle g \right\rangle + O(e^{-\frac{a}{\epsilon}(T-t)}) = O\left(\frac{\epsilon}{T-t}\right),
\end{equation}
uniformly in $0\leqslant t \leqslant T$ and $\epsilon>0$,\footnote{Uniformly here meaning that the constant the big-O implies is independent of the time variable $t$ and parameter $\epsilon$ within this range. We understand the estimate on the boundaries in the pointwise limit sense.}  where $a = |a_1| > 0$ is the spectral gap of $\mathcal{L}.$ 
\end{remark}
\begin{remark}
The unique invariant distribution $\Pi$ of $\boldsymbol{y}^1,$ whose existence we have postulated in (\textbf{H4}), is characterized as the solution to the PDE
$$
\begin{cases}
\mathcal{L}^* \Pi = 0 \text{ in } \mathbb{R}^d, \\ 
\int_{\boldsymbol{y} \in \mathbb{R}^d} \Pi(dy) = 1,
\end{cases}
$$
where
$$
\mathcal{L}^* \Pi = \frac{1}{2}\tr\left[D^2 \left( \beta \beta^\intercal \Pi \right) \right] - D\left( \Pi \alpha \right),
$$
see \cite[Eq. (3.7)]{fouque2011multiscale}.
\end{remark}

\begin{example} \label{ex:MultiD_OU}
The multi-dimensional Ornstein-Uhlenbeck (OU) process
$$
d\boldsymbol{y}^1_t = \boldsymbol{\Lambda}\left( \boldsymbol{m} - \boldsymbol{y}^1_t \right) + \boldsymbol{\eta} d\boldsymbol{W}_t,
$$
where $\boldsymbol{\Lambda} \in \mathbb{R}^{d\times d}$ is diagonal, with positive entries, $\boldsymbol{m} \in \mathbb{R}^d,$ and the matrix $\boldsymbol{\eta} \in \mathbb{R}^{d\times d}$ is invertible, is such that all hypotheses (\textbf{H3})-(\textbf{H5}) are valid. In this case,
$$
\Pi(\boldsymbol{y}) = (2\pi)^{-d/2}\left(\det \boldsymbol{A} \right)^{-1/2}\exp\left(-\frac{1}{2}\left( \boldsymbol{y} - \boldsymbol{m} \right)^\intercal \boldsymbol{A}^{-1}\left( \boldsymbol{y} - \boldsymbol{m} \right) \right), 
$$
where $\boldsymbol{A}$ solves
$$
\boldsymbol{\Lambda} \boldsymbol{A} + \boldsymbol{A} \boldsymbol{\Lambda} = \boldsymbol{\eta}\boldsymbol{\eta}^\intercal.
$$
\end{example}

Now, going back to our problem \eqref{eq:OptControlProb}, let us observe that the optimal control in feedback form is
\begin{equation} \label{eq:FeedbackCharactOptmStrat1}
    \nu^*(t,q,\boldsymbol{y}) := \sgn\left( \partial_q J(t,q,\boldsymbol{y}) \right) \left( \frac{|\partial_q J(t,q,\boldsymbol{y})|}{(1+\phi)\kappa(\boldsymbol{y})} \right)^{\frac{1}{\phi}},
\end{equation}
see \cite{souza2021regularized}. As this reference shows, the ansatz $J(t,q,\boldsymbol{y}) = z(t,\boldsymbol{y})|q|^{1+\phi}$ yields for $z$ the PDE
\begin{equation} \label{eq:MainPDE}
  \begin{cases}    
    \partial_t z + \frac{1}{\epsilon}\mathcal{L}z + \phi\kappa^{-1/\phi}|z|^{1+1/\phi} - \gamma \sigma^{1+\phi} = 0,\\
    z|_{t=T} = -A.
  \end{cases}
\end{equation}

\subsection{Some previous results}

The work \cite{souza2021regularized} contains the proof of the following results concerning the solution $z$ of the PDE \eqref{eq:MainPDE}.
\begin{theorem} \label{thm:ProptsSolnMainPDE}
(a) There exists a unique continuous and bounded viscosity solution $z$ of \eqref{eq:MainPDE}. 

(b) The function $z$ satisfies
$$
\frac{1}{C(T-t+A^{-1/\phi})^{\phi}} \leqslant -z(t,\boldsymbol{y}) \leqslant \frac{C}{(T-t+A^{-1/\phi})^{\phi}} \hspace{1.0cm} \left( (t,\boldsymbol{y}) \in \left[0,T\right]\times\mathbb{R}^d \right),
$$
where the positive constant $C$ is independent of $A.$

(c) The value function $J$ satisfies $J(t,\,q,\,\boldsymbol{y}) = z(t,\boldsymbol{y})|q|^{1+\phi},$ for each $(t,q,\boldsymbol{y}) \in \left[0,T\right]\times \mathbb{R} \times \mathbb{R}^d.$
\end{theorem}

As a consequence of Theorem \ref{thm:ProptsSolnMainPDE} (c) and \eqref{eq:FeedbackCharactOptmStrat1}, we obtain the following characterization of the optimal strategy in terms of $z.$
\begin{corollary}
The optimal control $\nu^*$ in feedback form is given by
\begin{equation} \label{eq:OptContFeedbackWithAnsatz}
    \nu^*(t,q,\boldsymbol{y}) := - \left( - \frac{z(t,\boldsymbol{y})}{ \kappa(\boldsymbol{y})} \right)^{\frac{1}{\phi}}q.
\end{equation}
\end{corollary}

\subsection{Expanding in a power series in the fast mean-reversion parameter}

We formally expand
\begin{equation} \label{Expansion}
    z = \sum_{n=0}^\infty \epsilon^n z_n.
\end{equation}
Our aim is to find the zeroth and first-order terms in this expansion. In this direction, let us write
$$
\mathcal{L}_1 \left( v \right) := \partial_t v + \phi \kappa^{-1/\phi}\left|v\right|^{1+1/\phi} - \gamma \sigma^{1+\phi}.
$$ 
It follows that
\begin{equation} \label{FirstEqn}
    \mathcal{L}z_0 = 0,
\end{equation}
\begin{equation} \label{SecondEqn}
    \mathcal{L} z_1 + \mathcal{L}_1(z_0) = 0,
\end{equation}
\begin{equation} \label{ThirdEqn}
    \mathcal{L} z_2 + \mathcal{L}_1^\prime(z_0)\cdot z_1 = 0,
\end{equation}
and the terminal conditions ought to be $z_0|_{t=T} = -A $ and $z_1|_{t=T} = 0 = z_2|_{t=T}.$ In \eqref{ThirdEqn}, we have written
$$
\mathcal{L}_1^\prime(z_0)\cdot z_1 := \partial_t z_1 + \left(1+\phi\right)\left( \frac{\left|z_0\right|}{\kappa}\right)^{1/\phi} \sgn\left(z_0\right) z_1.
$$

\section{Leading-order approximation} \label{sec:OrderZero}

\subsection{Derivation of the leading-order approximation}

From \eqref{FirstEqn}, we derive $z_0 = z_0(t).$ Then, we obtain from \eqref{SecondEqn}, together with the corresponding terminal condition, that
\begin{equation} \label{eq:ODE_OrdZeroApprox}
    \begin{cases}
        0 = \left\langle \mathcal{L}_1 \left(z_0\right) \right\rangle = \partial_t z_0 + \phi \left\langle\kappa^{-1/\phi}\right\rangle\left|z_0\right|^{1+1/\phi} - \gamma \left\langle\sigma^{1+\phi}\right\rangle, \\
        z_0(T) = -A  ,
    \end{cases}
\end{equation}
whence we have the representation
\begin{equation} \label{eq:leadingOrdApprox}
z_0(t) = F^{-1}\left(T-t\right),\,0\leqslant t \leqslant T,
\end{equation}
where $F : \left[-A, -\left(\frac{\gamma\left\langle \sigma^{1+\phi}\right\rangle}{\phi\left\langle \kappa^{-1/\phi} \right\rangle} \right)^{\frac{\phi}{\phi+1}} \right[ \rightarrow \left[0,\infty \right[ $ is defined as
$$
F\left( \xi \right) := -\int_{-A}^{\xi} \frac{du}{\gamma \left\langle\sigma^{1+\phi}\right\rangle - \phi \left\langle\kappa^{-1/\phi}\right\rangle|u|^{1+1/\phi}}.
$$

\begin{definition} \label{def:Order0}
Our leading-order approximation $\overline{z}_0$ is given by
$$ 
\overline{z}_0 := z_0,
$$
where we described $z_0$ in \eqref{eq:leadingOrdApprox}.
% (b) Our leading-order approximation for the constrained problem is
% $$
% \overline{h}^\infty_0(t) := -\zeta \coth\left( \zeta(T-t) \right),
% $$
% where we introduced $\zeta$ in \eqref{eq:Param1LeadingOrder}.
\end{definition}

In the case $\phi = 1,$ we have the closed-form expression
\begin{equation} \label{eq:leadingOrdApprox_phiEqualToOne}
    z_0(t) = -\zeta\left( \frac{e^{2\zeta (T-t)} + \zeta^\prime}{e^{2\zeta (T-t)} - \zeta^\prime} \right),
\end{equation}
for the parameters
\begin{equation} \label{eq:Param1LeadingOrder_phiEqualToOne}
    \zeta := \sqrt{ \gamma \left\langle \sigma^2\right\rangle \left\langle \frac{1}{\kappa} \right\rangle }
\end{equation}
and 
\begin{equation} \label{eq:Param2leadingOrd_phiEqualToOne}
    \zeta^\prime = \frac{A + \sqrt{\gamma \frac{\left\langle \sigma^2\right\rangle}{ \left\langle \frac{1}{\kappa} \right\rangle}}}{A - \sqrt{\gamma\frac{ \left\langle \sigma^2\right\rangle}{ \left\langle \frac{1}{\kappa} \right\rangle} }}.
\end{equation}

\subsection{A first set of numerical experiments} \label{subsec:Numerics1}

In view of \eqref{eq:OptContFeedbackWithAnsatz}, our leading-order approximation $z_0 = z_0^\gamma$ of $z$ (cf. \eqref{eq:ODE_OrdZeroApprox} or \eqref{eq:leadingOrdApprox}), corresponding to a risk aversion parameter $\gamma,$ suggests us to use the rate of trading given in feedback form by
\begin{equation} \label{eq:strat_leading_order}
    \nu^{0,\gamma}(t,q,\boldsymbol{y}) = - \left( - \frac{z_0^\gamma(t)}{ \kappa(\boldsymbol{y})} \right)^{\frac{1}{\phi}}q.
\end{equation}
From here on, we will denote the inventory and cash processes corresponding to the strategy $\nu^{0,\gamma}$ by $Q^{0,\gamma}$ and $X^{0,\gamma},$ respectively. In particular, in the risk-neutral setting, we put
$$
Q^0 := Q^{0,\gamma = 0} \text{ and } X^0 := X^{0,\gamma = 0}.
$$
For any given benchmark strategy $\nu^b,$ we refer to the quantity (in basis points)
$$
\frac{X^{0,\gamma}_T + Q^{0,\gamma}_T\left(S_T - A |Q^{0,\gamma}_T|^\phi\sgn(Q^{0,\gamma}_T)\right) - \left[ X^b_T +Q^{b}_T\left(S_T - A |Q^{b}_T|^\phi\sgn(Q^{b}_T)\right)\right]}{X^b_T+Q^{b}_T\left(S_T - A |Q^{b}_T|^\phi\sgn(Q^{b}_T)\right)}\times 10^4,
$$
as the performance of $\nu^{0,\gamma}$ relative to $\nu^b$ (or simply the relative performance, when there is no ambiguity about the benchmark), where $X^b_T$ is the terminal cash we obtain from following strategy $\nu^b.$

Regarding the model dynamics, we take a two-dimensional OU process $\boldsymbol{y} = \left\{ \boldsymbol{y}_t = (y^{(1)}_{t},y^{(2)}_{t}) \right\}_t,$ 
$$
\begin{cases}
    B^1 := W^1,\, B^2 = \rho W^1 + \sqrt{1-\rho^2}W^2,\\
    \text{and } dy^{(i)}_{t} = \frac{\lambda_i}{\epsilon}\left( m_i - y^{(i)}_{t}\right)\,dt + \frac{\eta_i}{\sqrt{\epsilon}} dB^i_t \hspace{1.0cm} \left(i\in \left\{1,2\right\}\right),
\end{cases}
$$
for a two-dimensional Brownian motion $(W^1,\,W^2).$ Moreover, we assume that each factor models each one of the processes $\kappa$ and $\sigma:$
$$
\kappa_t = \kappa(y^{(1)}_t) \text{ and } \sigma_t = \sigma(y^{(2)}_t).
$$
Above, we take\footnote{We write $x\vee y := \max\left(x,\,y\right)$ and $x \wedge y := \min\left(x,\, y\right).$}
$$
\kappa(y) = \underline{\kappa} \vee \left( y \wedge \overline{\kappa} \right) \text{ and } \sigma(y) = \exp\left( \underline{\sigma} \vee \left( y \wedge \overline{\sigma} \right) \right),
$$
for the parameters $\underline{\kappa},\overline{\kappa},\underline{\sigma},$ and $\overline{\sigma}$ we present in Table \ref{tab:bounds_kappa_sigma}. 

% \usepackage{booktabs}
\begin{table}[]
\centering
\begin{tabular}{@{}cccc@{}}
\toprule
$\underline{\kappa}$ & $\overline{\kappa}$ & $\underline{\sigma}$ & $\overline{\sigma}$ \\ \midrule
$0.01$               & $1.1$               & $2$                  & $7$                 \\ \bottomrule
\end{tabular}
\caption{Bounds for stochastic liquidity and volatility.}
\label{tab:bounds_kappa_sigma}
\end{table}

We introduce the matrices
$$
\boldsymbol{\Lambda} = \begin{bmatrix}
\lambda_1 & 0\\
0 & \lambda_2
\end{bmatrix} \text{ and } \boldsymbol{\eta} = \begin{bmatrix}
\eta_1 & 0 \\
\rho \eta_2 & \sqrt{1-\rho^2}\eta_2
\end{bmatrix},
$$
implying 
$$
\boldsymbol{\eta}\boldsymbol{\eta}^\intercal = \begin{bmatrix}
\eta_1^2 & \rho\eta_1\eta_2 \\
\rho\eta_1\eta_2 & \eta_2^2
\end{bmatrix}.
$$
As a particular case of our Example \ref{ex:MultiD_OU}, let us recall that we have the closed-form expression
\begin{equation} \label{eq:Pi2DOUcase}
    \Pi(\boldsymbol{y}) = (2\pi)^{-1}\left(\det \boldsymbol{A} \right)^{-1/2}\exp\left(-\frac{1}{2}\left( \boldsymbol{y} - \boldsymbol{m} \right)^\intercal \boldsymbol{A}^{-1}\left( \boldsymbol{y} - \boldsymbol{m} \right) \right), 
\end{equation}
with $\boldsymbol{A}$ being the solution to the matrix equation $\boldsymbol{\Lambda} \boldsymbol{A} + \boldsymbol{A} \boldsymbol{\Lambda} = \boldsymbol{\eta}\boldsymbol{\eta}^\intercal,$ which here is given explicitly by
\begin{equation} \label{eq:MatrixA}
    \boldsymbol{A} = \begin{bmatrix}
        \frac{\eta_1^2}{2\lambda_1} & \frac{\rho \eta_1 \eta_2}{\lambda_1 + \lambda_2} \\
        \frac{\rho \eta_1 \eta_2}{\lambda_1 + \lambda_2} & \frac{\eta_2^2}{2\lambda_2}
    \end{bmatrix}.
\end{equation}

In Section \ref{subsec:StochVolAndLiq}, we obtained our estimates $\widehat{\phi},\,\widehat{\lambda}_i,\,\widehat{m}_i,\, \widehat{\eta}_i,$ and $\widehat{\rho}$ of the parameters $\phi,\, \lambda_i^\epsilon := \lambda_i/\epsilon,\, m_i,\, \eta_i^\epsilon := \eta_i/\sqrt{\epsilon}$ (for $i\in \left\{1,2\right\}),$ and $\rho,$ respectively, for BTCUSDT at December 19, 2022. We expose the simulation's parameters resulting from those developments in Tables \ref{tab:params_simulations} and \ref{tab:paramaters_sim_2}. We provide an example of sample paths for the temporary price impact coefficient, the log-volatility, and the corresponding price process in Figure \ref{fig:sample_paths_sim}.

\begin{table}[]
\centering
\begin{tabular}{@{}cccc@{}}
\toprule
$i$ & $\lambda_i$ & $m_i$    & $\eta_i$  \\ \midrule
$1$ & $1905.2180$ & $0.3782$ & $4.0134$  \\
$2$ & $1279.7954$ & $4.7810$ & $19.0326$ \\ \bottomrule
\end{tabular}
\caption{Parameters we use to model our two-dimensional Markov diffusion driving the stochastic liquidity and volatility.}
\label{tab:params_simulations}
\end{table}

% \usepackage{booktabs}
\begin{table}[]
\centering
\begin{tabular}{@{}ccc@{}}
\toprule
$\phi$ & $\rho$ & $\epsilon$ \\ \midrule
0.2833 & 0.2096 & 0.0008     \\ \bottomrule
\end{tabular}
\caption{Remaining parameters we fix in our simulations, rounded to four decimal places.}
\label{tab:paramaters_sim_2}
\end{table}

\begin{figure}
    \centering
    \includegraphics[scale=0.4]{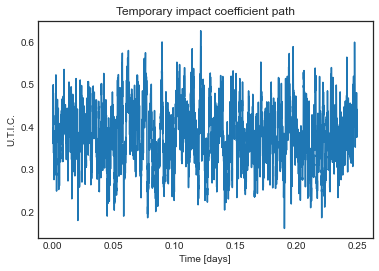}
    \includegraphics[scale=0.4]{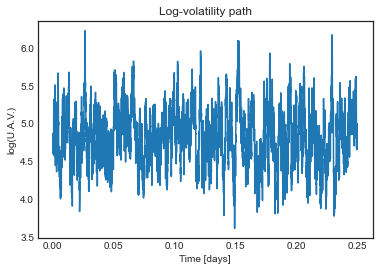}
    \includegraphics[scale=0.4]{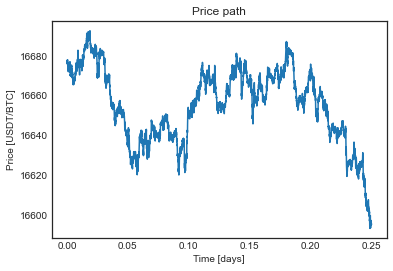}
    \caption{A sample path for the temporary impact coefficient, for the log-volatility, and for the corresponding price path.}
    \label{fig:sample_paths_sim}
\end{figure}

We will take $\epsilon = \min\left(\widehat{\lambda}_1,\,\widehat{\lambda}_2\right)^{-1},$ see Table \ref{tab:paramaters_sim_2}. We assess the performance of the strategies we described above by running $10^4$ Monte Carlo simulations for each set of parameters, using the same  price, volatility, and liquidity innovations across experiments for distinct strategies (the approximations and its benchmarks), with $T = 0.25$ day, and we provide our results in what follows. We consider a liquidation program with initial data given in Table \ref{tab:init_condts}. During the restricted time window from $8$ AM to $2$ PM for BTCUSDT on Binance at the day we analyzed it, the traded volume amounted to 33217.79 BTC. Thus, we are simulating an execution of roughly one third of that window's traded volume.

% \usepackage{booktabs}
\begin{table}[]
\centering
\begin{tabular}{@{}ccccc@{}}
\toprule
$X_0\,{[}\${]}$ & $Q_0$ {[}BTC{]} & $S_0$ {[}$\$$/BTC{]} & $y^{(1)}_0$ {[}U.T.I.C{]} & $y^{(2)}_0$ {[}log(U.A.V.){]} \\ \midrule
$0$              & $10000$         & 16676                   & $0.3782$                  & $4.7810$                      \\ \bottomrule
\end{tabular}
\caption{Initial conditions for our simulations.}
\label{tab:init_condts}
\end{table}

We benchmark the performance of our leading-order approximation using the standard Almgren-Chriss $\nu^{AC}$ in which we assume $\kappa \equiv m_1$ and $\sigma \equiv e^{m_2}.$ Thus, in feedback form, $\nu^{AC}(t,q) = -\left( - z^{AC}(t)/m_1 \right)^{1/\phi} q,$ where $z^{AC}$ solves \eqref{eq:ODE_OrdZeroApprox} with $m_1$ and $e^{m_2}$ in place of $\kappa$ and $\sigma,$ respectively. Thus, as our benchmark, we consider the strategy a trader assuming constant liquidity and volatility (equal to their long-run mean values) should use if she were to optimize with our objective functional.

Firstly, we analyze the risk neutral setting, i.e., with $\gamma = 0.$ We present in Figure \ref{fig:risk_neutral_variables} a sample path of the state and control variables for both the approximation and the benchmark corresponding to the innovations of the paths we showed in Fig \ref{fig:sample_paths_sim}. Even in the risk neutral setting, the leading-order approximation has the advantage of being adaptative with respect to the stochastic liquidity, see \eqref{eq:leadingOrdApprox}. In Table \ref{tab:risk_neutral_perf}, we present some quantitative aspects to support that the approximation not only consistently outperforms the benchmark, but it also commonly ends up holding less inventory. Thus, we can conclude that it provides a considerable edge from the viewpoint of a trader willing to liquidate her sizeable portfolio. In Figure \ref{fig:histograms_risk_neutral}, we further illustrate the approximation's performance.

\begin{figure}
    \centering
    \includegraphics[scale=0.4]{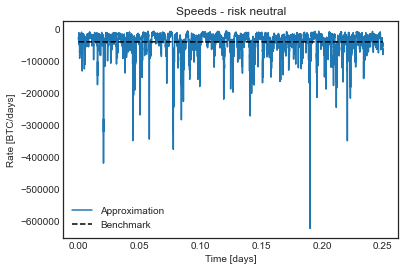}
    \includegraphics[scale=0.4]{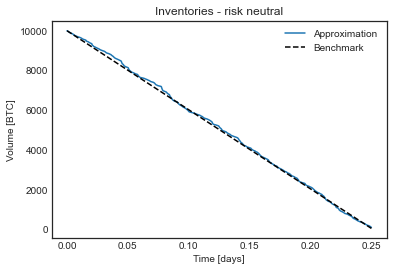}
    \includegraphics[scale=0.4]{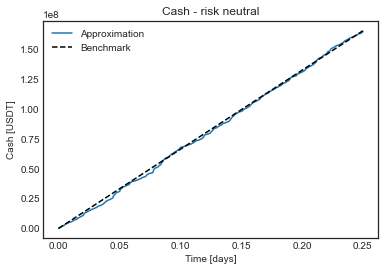}
    \caption{Time evolution of the control and state variables for the trader.}
    \label{fig:risk_neutral_variables}
\end{figure}

% \usepackage{booktabs}
\begin{table}[]
\centering
\begin{tabular}{@{}cccccc@{}}
\toprule
Relative performance [bps] & Improvement rate & $\widehat{\mathbb{P}}(X^0_T > X^b_T)$ & $\widehat{\mathbb{P}}(Q^0_T < X^b_T)$ & $\frac{\widehat{\mathbb{E}}\left[X^0_T\right]}{Q_0S_0}$ & $\frac{\widehat{\mathbb{E}}\left[Q^0_T\right]}{Q_0}$ \\ \midrule
$0.3580$             & $68.43 \%$       & $0.6186$                    & $0.6136$                    & $99.49\%$                               & $0.47\%$                             \\
$(0.9593)$           &                  &                             &                             & $(0.37\%)$                              & $(0.29\%)$                           \\ \bottomrule
\end{tabular}
\caption{Some performance indicators for the risk neutral leading-order approximation. The benchmark is its AC counterpart using the long-term mean values of stochastic liquidity and volatility as their constant parameter values. We refer to the ``improvement rate'' as the number of simulations in which $\nu^0$ improves $\nu^b$ in terms of relative performance divided by the total number of simulations. The probability measure $\widehat{\mathbb{P}}$ above is the empirical one, and $\widehat{\mathbb{E}}$ is taken with respect to it. Quantities within parentheses are standard deviations corresponding to the averages immediately above them.}
\label{tab:risk_neutral_perf}
\end{table}

\begin{figure}
    \centering
    \includegraphics[scale=0.4]{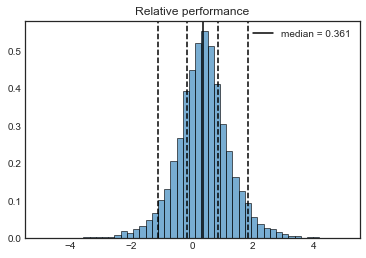}
    \includegraphics[scale=0.4]{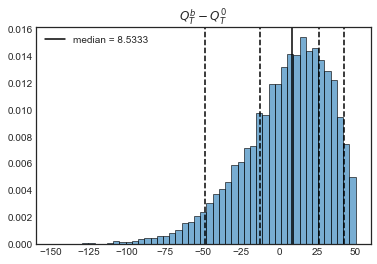}
    \caption{Histograms of relative performance and the inventory difference $Q^b-Q^0$ between those corresponding to the risk neutral ($\gamma=0$) versions of the benchmark $\nu^{AC}$ and leading-order approximation $\nu^0,$ respectively.}
    \label{fig:histograms_risk_neutral}
\end{figure}

Next, we move the risk averse case, i.e., $\gamma > 0.$ Our benchmark is still $\nu^{AC},$ with the same $\gamma$ as we regard in the approximation. Regarding its behavior, we showcase in Figure \ref{fig:risk_av_state_vars} its inventory and cash. As we would expect, it is more front loaded than the risk neutral one. Moreover, relative to the benchmark, by using higher values of risk aversion the trader still maintains a consistent average edge over the benchmark with the same risk aversion level, see Table \ref{tab:risk_averse_rel_perf}. However, as we can see in Figure \ref{fig:perf_risk_averse}, the performance relative to the benchmark becomes a bit more uncertain, as its variance is a bit larger. We point out that the risk averse agent is not directly adaptative with respect to the price volatility. The trader is only sensitive this process through a suitable effective mean, see \eqref{eq:ODE_OrdZeroApprox} and \eqref{eq:strat_leading_order}. We have illustrated how using the effective mean adds up some value relative to using $\sigma \equiv e^{m_2}.$

\begin{figure}
    \centering
    \includegraphics[scale=0.4]{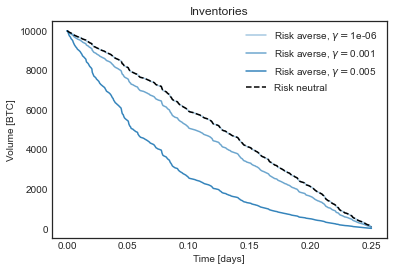}
    \includegraphics[scale=0.4]{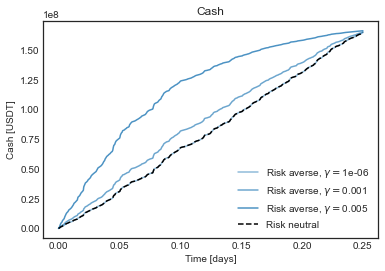}
    \caption{Some state variables illustrating the risk averse setting.}
    \label{fig:risk_av_state_vars}
\end{figure}

% \usepackage{booktabs}
\begin{table}[]
\centering
\begin{tabular}{@{}cc@{}}
\toprule
$\gamma$ & Relative performance [bps] \\ \midrule
$0.000001$ & $0.3580$                     \\
         & $(0.9594)$                   \\ \midrule
$0.001$ & $0.3567$                     \\
         & $(1.0983)$                   \\ \midrule
$0.005$ & $0.2840$                     \\
         & $(1.6879)$                   \\ \bottomrule
\end{tabular}
\caption{Relative performance of the risk averse leading-order approximation with respect to the risk averse AC counterpart. Within parentheses are the standard deviations of the averages immediately above them.}
\label{tab:risk_averse_rel_perf}
\end{table}

\begin{figure}
    \centering
    \includegraphics[scale=0.4]{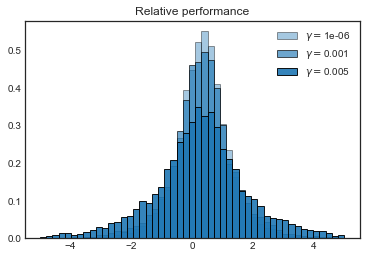}
    \caption{Relative performance of the risk averse leading-order approximation with respect to the risk averse AC counterpart.}
    \label{fig:perf_risk_averse}
\end{figure}

Finally, relative to the risk neutral setting, adding some risk aversion seems to be useful in any case, see Table \ref{tab:risk_av_v_risk_neutral}. In effect, the penalization of the running inventory does not always overcome the terminal wealth. Thus, adding at a small degree of urgency to her program (i.e., being weakly risk averse), leads to ending up with a larger terminal wealth on average than the risk neutral counterpart, but with inventory and price risk mitigated. Furthermore, larger values of risk aversion mitigate the variance of the terminal inventory, and makes its mean value be closer to zero, see Figure \ref{fig:inventories_risk_averse}.  

% \usepackage{booktabs}
\begin{table}[]
\centering
\begin{tabular}{@{}cccc@{}}
\toprule
$\gamma$ & Relative performance [bps] & $\widehat{\mathbb{P}}\left(X^{0,\gamma}_T > X^{0,\gamma=0}_T\right)$ & $\widehat{\mathbb{P}}\left(Q^{0,\gamma}_T < Q^{0,\gamma=0}_T\right)$ \\ \midrule
$0.000001$ & $0.00002$                    & $0.9743$                                                     & $0.9996$                                                    \\
         & $(0.0027)$                   &                                                            &                                                            \\
$0.001$ & $-0.0126$                    & $0.9710$                                                     & $0.9996$                                                     \\
         & $(2.5755)$                   &                                                            &                                                            \\
$0.005$ & $-0.4178$                    & $0.9580$                                                     & $0.9996$                                                     \\
         & $(10.0674)$                  &                                                            &                                                            \\ \bottomrule
\end{tabular}
\caption{Risk averse performance relative to the risk neutral setting for the leading-order approximation. We denoted by $\widehat{\mathbb{P}}$ the empirical probability measure.}
\label{tab:risk_av_v_risk_neutral}
\end{table}

\begin{figure}
    \centering
    \includegraphics[scale=0.4]{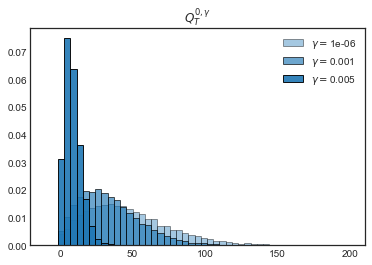}
    \caption{Inventory histograms for the risk averse leading-order approximations. The probability measure $\widehat{\mathbb{P}}$ above is the empirical one.}
    \label{fig:inventories_risk_averse}
\end{figure}

\section{First-order Correction} \label{sec:OrderOne}

\subsection{Derivation of the first-order correction}

From \eqref{SecondEqn}, we have
\begin{align} \label{eq:PDE_OrdOneApprox}
    \begin{split}
        \mathcal{L}z_1 &= -\mathcal{L}_1\left(z_0\right) \\
        &= -\left( \mathcal{L}_1 - \left\langle \mathcal{L}_1 \right\rangle \right) \left( z_0 \right) \\
        &= \phi\left|z_0\right|^{1+1/\phi}\left( \left\langle \kappa^{-1/\phi} \right\rangle - \kappa^{-1/\phi} \right) + \gamma \left(\sigma^{1+\phi} - \left\langle \sigma^{1+\phi} \right\rangle \right).
    \end{split}
\end{align}
Upon introducing the solutions $\varphi_0$ and $\varphi_1$ of the Poisson equations
\begin{equation} \label{eq:Phi0}
    \begin{cases}
        \mathcal{L}\varphi_0 = \left\langle \kappa^{-1/\phi} \right\rangle - \kappa^{-1/\phi}, \text{ in } \mathbb{R}^d, \\
        \left\langle \varphi_0 \right\rangle = 0,
    \end{cases}
\end{equation}
and
\begin{equation} \label{eq:Phi1}
    \begin{cases}
        \mathcal{L}\varphi_1 = \sigma^{1+\phi} - \left\langle \sigma^{1+\phi} \right\rangle , \text{ in } \mathbb{R}^d, \\
        \left\langle \varphi_1 \right\rangle = 0,
    \end{cases}
\end{equation}
we infer that
$$
\mathcal{L}\left( z_1 - \phi \left|z_0\right|^{1+1/\phi} \varphi_0 - \gamma \varphi_1 \right) = 0,
$$
whence
$$
z_1 =  \phi \left|z_0\right|^{1+1/\phi}\varphi_0 + \gamma \varphi_1 + c,
$$
where $c=c(t).$ We determine $c$ from \eqref{ThirdEqn}, since this equation implies that 
$$
\left\langle \mathcal{L}_1^\prime\left(z_0\right)\cdot z_1 \right\rangle = 0,
$$
in such a way that
\begin{equation} \label{eq:ODE_bdry_layer}
    c^\prime +  b_0 c + b_1  = 0,
\end{equation}
where
$$
b_0(t) := (1+\phi)\left| z_0 \right|^{1/\phi}\sgn\left( z_0 \right) \left\langle \kappa^{-1/\phi} \right\rangle = -(1+\phi)\left| z_0 \right|^{1/\phi} \left\langle \kappa^{-1/\phi} \right\rangle
$$
and
\begin{align*}
    b_1(t) :=&\, (1+\phi)|z_0|^{ 1/\phi} \left[ \phi z_0|z_0|^{ 1/\phi} \left\langle \frac{\varphi_0}{\kappa^{1/\phi}} \right\rangle +\gamma \sgn\left(z_0\right) \left\langle \frac{\varphi_1}{\kappa^{1/\phi}} \right\rangle \right] \\
    =&\, (1+\phi)|z_0|^{ 1/\phi} \left[\phi z_0|z_0|^{ 1/\phi} \left\langle \frac{\varphi_0}{\kappa^{1/\phi}} \right\rangle - \gamma \left\langle \frac{\varphi_1}{\kappa^{1/\phi}} \right\rangle \right].
\end{align*}
Moreover, we stipulate $\left\langle z_1\left(T,\cdot \right) \right\rangle =0,$ which implies $c(T) =0.$ Thus,
\begin{equation} \label{eq:BdryLayer}
    c(t) = \int_t^T e^{\int_t^u b_0(r)\,dr} b_1(u)\,du. 
\end{equation}
\begin{definition}
We define the first-order correction as
$$
\overline{z}_1 := \overline{z}_0 + \epsilon \left(\phi \left|z_0\right|^{1+1/\phi}\varphi_0 + \gamma\varphi_1 + c\right),
$$
where $\varphi_0$ and $\varphi_1$ are given by \eqref{eq:Phi0} and \eqref{eq:Phi1}, respectively, and we presented the boundary layer $c$ in \eqref{eq:BdryLayer}.
% (b) Our first-order correction for the constrained problem is
% $$
% \overline{h}^\infty_1 := \overline{h}^\infty_0 + \epsilon \left[ \varphi_0 + \left(\overline{h}^\infty_0\right)^2 \varphi_1 + c^\infty\right],
% $$
% where $\overline{h}_0^\infty$ is described in Definition \ref{def:Order0} (b), $\varphi_0$ and $\varphi_1$ are as in (a), and 
% $$
% c(t) = 2\int_t^T e^{2\left\langle \frac{1}{\kappa} \right\rangle\int_t^u \overline{h}_0^\infty(\tau)\,d\tau } \left[ \left\langle \frac{\varphi_0}{\kappa} \right\rangle + \left\langle \frac{\varphi_1}{\kappa} \right\rangle \left( \overline{h}_0^\infty(u) \right)^2 \right]\overline{h}_0^\infty(u)\,du.
% $$
\end{definition}

% As in Lemma \ref{lem:ConvApproxOrd0}, we can prove a convergence result for the first-order corrections.
% \begin{lemma} \label{lem:ConvApproxOrd1}
% The convergence
% $$
% \overline{h}_1 \xrightarrow{A \rightarrow \infty} \overline{h}^\infty_1
% $$
% holds compactly $\left[ 0,T \right[\times \mathbb{R}^d.$
% \end{lemma}

\subsection{Some complementary numerical experiments} \label{subsec:Numerics2}

We will consider the same particular model as in Section \ref{subsec:Numerics1}. Our first-order correction $z_1 = z_1^\gamma$ (corresponding to a risk aversion parameter $\gamma$) leads us to consider the strategy $\nu^{1,\gamma}$ given in feedback form by
$$
\nu^{1,\gamma}(t,q,\boldsymbol{y}) = - \left( - \frac{\overline{z}_1^\gamma(t)}{ \kappa(\boldsymbol{y})} \right)^{\frac{1}{\phi}}q.
$$
We write $\nu^1 \equiv \nu^{1,\gamma}$ whenever there is no danger of confusion, and we denote the cash and inventory processes resulting from following this trading rate by $X^{1,\gamma}$ and $Q^{1,\gamma},$ respectively. However, we observe that we must address some numerical subtleties concerning the computation of $\varphi_0$ and $\varphi_1$ in \eqref{eq:Phi0} and \eqref{eq:Phi1}, respectively. We present a rigorous analysis of this issue in \ref{app:LeadingOrder}. In the present section, we resort to discuss the results we obtain by implementing this first-order correction term to the leading-order approximation. We run $10000$ Monte Carlo simulations, with the same price, volatility, and impact innovations as in Section \ref{subsec:Numerics1}.

We present in Table \ref{tab:rel_perf_ord1_to_ord0} the relative performance of $\nu^1$ with respect to $\nu^0.$ We see that $\nu^1$ does perform consistently better than $\nu^0,$ but also that that the edge is quite marginal. It is noteworthy that the relative performance of $\nu^1$ with respect to $\nu^0$ for moderately risk averse traders is of the order of $\epsilon,$ see Table \ref{tab:paramaters_sim_2}. We note that $\nu^1$ is adaptative with respect to the factor $y^{(2)},$ whence for high risk aversion levels it ends up not improving the terminal wealth of $\nu^0$, but rather concentrating in managing inventory risk throughout the execution program. Moreover, $\nu^1$ typically liquidates more inventory and terminates with more cash than $\nu^0,$ which further shows its enhanced effectiveness from a risk management viewpoint. 

% \usepackage{booktabs}
\begin{table}[]
\centering
\begin{tabular}{@{}cc@{}}
\toprule
$\gamma$ & Relative performance [bps] \\ \midrule
$0$      & $0.000819$                   \\
         & $(0.072848)$                 \\ \midrule
$0.000001$ & $0.000818$                   \\
         & $(0.072818)$                 \\ \midrule
$0.001$    & $0.000336$                   \\
         & $(0.047897)$                 \\ \midrule
$0.005$    & $-0.000373$                  \\
         & $(0.009375)$                 \\ \bottomrule
\end{tabular}
\caption{Relative performance of the first-order correction $\nu^1$ relative to the leading-order couterpart $\nu^0$. }
\label{tab:rel_perf_ord1_to_ord0}
\end{table}

% \usepackage{booktabs}
\begin{table}[]
\centering
\begin{tabular}{@{}ccc@{}}
\toprule
$\gamma$ & $\widehat{\mathbb{P}}\left( X^{1,\gamma}_T > X^{0,\gamma}_T \right)$ & $\widehat{\mathbb{P}}\left( Q^{1,\gamma}_T < Q^{0,\gamma}_T \right)$ \\ \midrule
$0$        & $0.9989$                                                      & $0.9996$                                                      \\
$0.000001$ & $0.9989$                                                      & $0.9996$                                                     \\
$0.001$    & $0.9985$                                                      & $0.9995$                                                      \\
$0.005$    & $0.9972$                                                      & $0.9994$                                                      \\ \bottomrule
\end{tabular}
\caption{Terminal cash and inventory comparison among the first-order correction and leading-order approximation. Above, $\widehat{\mathbb{P}}$ denotes the empirical probability measure.}
\label{tab:term_cash_and_inv_ord1_to_ord0}
\end{table}

For the innovations leading to the paths we showed in Figure \ref{fig:sample_paths_sim}, we present the difference of $\nu^1$ and $\nu^0$ in terms of $\nu^0$ itself for different risk aversion levels in Figure \ref{fig:speed_rel_diff_ord1_to_ord0}. As the first-order correction is more front-loaded than the leading-order one, it decelerates relative to $\nu^0$ towards the end, cf. Table \ref{tab:term_cash_and_inv_ord1_to_ord0}. Since we are working in a very fast mean reverting market, see Table \ref{tab:fast_mr_params}, the correction we make in the first-order approximation is not too sizeable. Since the bulk of the edge added over the other benchmarks we considered were already present in the leading-order approximation, the fact that the latter is much easier to compute advocates in favor of prioritizing its use in practice.\footnote{Of course, we can only conclude this after computing the first-order correction and comparing to the leading-order one.}

\begin{figure}
    \centering
    \includegraphics[scale=0.4]{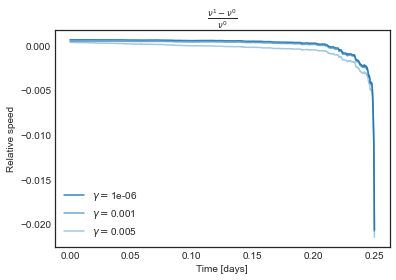}
    \caption{Difference $\nu^1-\nu^0$ relative to $\nu^0$ for various values of $\gamma.$ As $\nu^1$ is more front loaded than $\nu^0$, it slows down relative to it by the end of the execution program (recall $\nu^0 < 0$).}
    \label{fig:speed_rel_diff_ord1_to_ord0}
\end{figure}

\section{Accuracy results} \label{sec:accuracy}

\subsection{Leading-order approximations}

\begin{theorem} \label{thm:AsymptoticnessOrd0}
We have\footnote{Here and henceforth, the big-Os may depend on the point $\boldsymbol{y}.$}
$$
z(t, \boldsymbol{y}) = \overline{z}_0(t, \boldsymbol{y}) + O(\epsilon), \text{ as } \epsilon \downarrow 0,
$$
for each $(t,\boldsymbol{y}) \in \left[0,T\right]\times \mathbb{R}^d.$
\end{theorem}
\begin{proof}
Throughout this proof, all the generic constants $C$ figuring in the estimates are independent of $\epsilon.$ Let us write 
\begin{equation} \label{Residue1}
    R_0:= z-\overline{z}_0 - \epsilon z_1,
\end{equation}
From Eqs. \eqref{eq:MainPDE}, \eqref{eq:ODE_OrdZeroApprox} and \eqref{eq:PDE_OrdOneApprox}, we deduce
\begin{equation} \label{eq:PDEofR0}
    \partial_t R_0 + \frac{1}{\epsilon}\mathcal{L} R_0 +\phi\kappa^{-1/\phi}G(z,z_0)R_0 + \epsilon\left[ G(z,z_0)z_1 + \partial_t z_1 \right] = 0,
\end{equation}
with terminal condition
$$
R(T,\boldsymbol{y}) = -\epsilon z_1(T,\boldsymbol{y}),
$$
where we have written
$$
G(s,s^\prime) := \begin{cases}
\frac{|s|^{1+1/\phi} - |s^\prime|^{1+1/\phi}}{s-s^\prime} &\text{ if } s\neq s^\prime,\\
\left(1 + \frac{1}{\phi}\right)|s|^{1/\phi}\sgn(s) &\text{ otherwise. }
\end{cases}
$$
Particularly, we have 
\begin{equation} \label{eq:basic_G_bound}
    | G(s,s^\prime) | \leq C \left( |s|^{1/\phi} + \left| s^\prime \right|^{1/\phi} \right).
\end{equation}

For each $\left(t,\boldsymbol{y}\right) \in \left[0,T\right]\times\mathbb{R}^d,$ let us set 
$$
a\left(t,\boldsymbol{y}\right) := \phi\kappa\left( \boldsymbol{y} \right)^{-1/\phi}G\left(z\left(t,\boldsymbol{y}\right),\,z_0(t)\right) \text{ and } f\left(t,\boldsymbol{y}\right) := G\left(z(t,\boldsymbol{y}),z_0(t)\right)z_1(t,\boldsymbol{y}) + \partial_t z_1(t,\boldsymbol{y}) 
$$
Firstly, we have
\begin{align} \label{eq:bdd_of_coef_a_leading_ord_approx}
  \begin{split}
    \left| a(t,\,\boldsymbol{y}) \right| &\leq C | G\left(z\left(t,\boldsymbol{y}\right),\,z_0(t)\right) | \\
    &\leq C\left( \left|z\left(t,\boldsymbol{y}\right)\right|^{1/\phi} + \left| z_0(t) \right|^{1/\phi} \right) \\
    &\leq C,
  \end{split}
\end{align}
since $z$ and $z_0$ are bounded (independently of $\epsilon,$ cf. Theorem \ref{thm:ProptsSolnMainPDE} (b) and \eqref{eq:leadingOrdApprox}). Secondly, it follows analogously that
\begin{equation} \label{eq:init_est_src_acc_leading_order}
    \left| f(t,\,\boldsymbol{y}) \right| \leq C\left( \left| z_1\left( t,\,\boldsymbol{y} \right) \right| + \left| \partial_t z_1\left( t,\,\boldsymbol{y} \right) \right| \right).
\end{equation}
From the boundedness properties of $z_0$ and $c$ (cf. \eqref{eq:leadingOrdApprox} and \eqref{eq:BdryLayer}), it follows that 
\begin{equation} \label{eq:bddnss_z1}
    \left| z_1\left( t,\,\boldsymbol{y} \right) \right| \leq C\left( 1 + \left| \varphi_0(\boldsymbol{y}) \right| + \left| \varphi_1(\boldsymbol{y}) \right| \right).
\end{equation}
From \eqref{eq:leadingOrdApprox} and \eqref{eq:BdryLayer}, alongside the ODEs \eqref{eq:ODE_OrdZeroApprox} and \eqref{eq:ODE_bdry_layer}, we also see that
\begin{equation} \label{eq:bddnss_z0prime}
    |z_0^\prime(t)| \leq C \text{ and } |c^\prime(t)| \leq C \hspace{1.0cm} (0 \leq t \leq T).
\end{equation}
Furthermore, we also see from \eqref{eq:leadingOrdApprox} that $|z_0| = -z_0 \geq C > 0,$ from where we deduce the estimate 
$$
\left| \partial_t z_1\left(t,\,\boldsymbol{y}\right) \right| \leq C \left( 1 + \left| \varphi_0(\boldsymbol{y}) \right| + \left| \varphi_1(\boldsymbol{y}) \right| \right),
$$
whence, together with \eqref{eq:init_est_src_acc_leading_order}, we derive the following   
\begin{equation} \label{eq:bddns_src_f_leading_order_approx}
    \left| f(t,\,\boldsymbol{y}) \right| \leq C \left( 1 + \left| \varphi_0(\boldsymbol{y}) \right| + \left| \varphi_1(\boldsymbol{y}) \right| \right) \leq C \left( 1 + |\boldsymbol{y}|^n \right),
\end{equation}
for a sufficiently large integer $n \geq 1,$ since $\varphi_0$ and $\varphi_1$ are at most polynomially growing.

% From the fact that $z_1$ is at most polynomially growing in the spatial variable $\boldsymbol{y},$ uniformly on time $t \in \left[0,T\right],$ we deduce that the same is valid for $f.$ 

Using the Feynman-Kac Theorem, we infer the representation
\begin{equation} \label{eq:SecToLastEstToConcludeOrdZeroApprox}
    R_0\left(t,\,\boldsymbol{y}\right) = \epsilon\mathbb{E}_{t,\boldsymbol{y}}\left[\int_t^T e^{\int_t^u a(r,\boldsymbol{y}_r)\,dr}f(u,\boldsymbol{y}_u) \,du - e^{\int_t^T a(r,\boldsymbol{y}_r)\,dr}z_1(T,\boldsymbol{y}_T) \right].
\end{equation}
From \eqref{eq:bdd_of_coef_a_leading_ord_approx}, \eqref{eq:bddnss_z1}, \eqref{eq:bddns_src_f_leading_order_approx}, and the identity \eqref{eq:SecToLastEstToConcludeOrdZeroApprox}, we see that 
$$
\left| R_0\left(t,\boldsymbol{y}\right) \right| \leq C \epsilon\mathbb{E}_{t,\boldsymbol{y}}\left[\int_t^T \left( 1 +  |\boldsymbol{y}_u|^n\right) \,du + 1 + |\boldsymbol{y}_T|^n \right] \leq C \epsilon \left( 1 + \sup_{u \geq 0} \mathbb{E}_{0,\boldsymbol{y}}\left[ |\boldsymbol{y}^1_u|^n \right] \right),
$$
which, according to our hypothesis (\textbf{H4}), implies
$$
\left| R_0\left(t,\boldsymbol{y}\right) \right| \leqslant C\epsilon,
$$
thus finishing the proof.
\end{proof}

\subsection{First-order correction}

\begin{lemma} \label{lem:z2}
We can take $z_2$ satisfying \eqref{ThirdEqn}, and at most polynomially growing in the spatial variable.
\end{lemma}
\begin{proof}
    A function such as $z_2$ must comply with the subsequent equations:
    \begin{align} \label{eq:NecOrigRelation_z2}
      \begin{split}
        0 &= \mathcal{L}z_2 + \mathcal{L}_1^\prime(z_0)\cdot z_1 \\
          &= \mathcal{L}z_2 + \mathcal{L}_1^\prime(z_0)\cdot z_1 - \left\langle \mathcal{L}_1^\prime(z_0)\cdot z_1  \right\rangle \\
          &= \mathcal{L}z_2 + \left( \partial_t z_1 - \left\langle \partial_t z_1 \right\rangle \right) + \left[ (1+\phi)\left( \frac{|z_0|}{\kappa} \right)^{1/\phi} \sgn(z_0)z_1 - \left\langle (1+\phi)\left( \frac{|z_0|}{\kappa} \right)^{1/\phi} \sgn(z_0)z_1 \right\rangle \right] \\
          &= \mathcal{L}z_2 + \partial_t(z_1 - c) + (1+\phi)|z_0|^{1/\phi} \sgn(z_0)\left( \frac{z_1}{\kappa^{1/\phi}} - \left\langle \frac{z_1}{\kappa^{1/\phi}} \right\rangle \right).
      \end{split}
    \end{align}
    Firstly, we notice that
    \begin{equation}\label{eq:defn_of_little_d}
        \partial_t(z_1 - c) = (\phi+1)|z_0|^{1/\phi}\sgn(z_0)\varphi_0.
    \end{equation}
    Secondly, we write
    \begin{equation} \label{eq:rewriting_z1_by_powerkappa}
        \frac{z_1}{\kappa^{1/\phi}} - \left\langle \frac{z_1}{\kappa^{1/\phi}} \right\rangle = \phi|z_0|^{1+1/\phi} \left( \frac{\varphi_0}{\kappa^{1/\phi}} - \left\langle \frac{\varphi_0}{\kappa^{1/\phi}} \right\rangle \right) + \gamma \left( \frac{\varphi_1}{\kappa^{1/\phi}} - \left\langle \frac{\varphi_1}{\kappa^{1/\phi}} \right\rangle \right) + c\left( \frac{1}{\kappa^{1/\phi}} - \left\langle \frac{1}{\kappa^{1/\phi}} \right\rangle \right).
    \end{equation}
    In this way, by introducing the auxiliary function
    \begin{equation} \label{eq:aux_b2_defn}
        b_2(t) := -(\phi + 1)|z_0(t)|^{1/\phi}\sgn(z_0(t)) \hspace{1.0cm} (0 \leqslant t \leqslant T),
\end{equation}
    as well as the at most polynomially growing solutions $\varphi_2,\,\varphi_3,$ and $\varphi_4$ to the Poisson equations
    $$
    \begin{cases}
        \mathcal{L}\varphi_2 = -\varphi_0,\\
        \mathcal{L}\varphi_3 = \left\langle \frac{\varphi_0}{\kappa^{1/\phi}} \right\rangle - \frac{\varphi_0}{\kappa^{1/\phi}},\\
        \mathcal{L}\varphi_4 = \left\langle \frac{\varphi_1}{\kappa^{1/\phi}} \right\rangle - \frac{\varphi_1}{\kappa^{1/\phi}},\\
        \left\langle \varphi_2 \right\rangle = \left\langle \varphi_3 \right\rangle = \left\langle \varphi_4 \right\rangle = 0,
    \end{cases}
    $$
    it follows from \eqref{eq:Phi0}, \eqref{eq:NecOrigRelation_z2}, \eqref{eq:defn_of_little_d}, and \eqref{eq:rewriting_z1_by_powerkappa} that $z_2$ must satisfy
    $$
    \mathcal{L}\left( z_2 - \left\{ b_2 \, \varphi_2 + \phi\, b_2 \, |z_0|^{1+1/\phi}\varphi_3 + \gamma\, b_2\, \varphi_4 - c\, b_2\, \varphi_0 \right\} \right) = 0. 
    $$
    Hence, upon taking 
    \begin{equation} \label{eq:particular_z2}
            z_2 =  b_2 \, \varphi_2 + \phi\, b_2 \, |z_0|^{1+1/\phi}\varphi_3 + \gamma\, b_2\, \varphi_4 - c\, b_2\, \varphi_0,
    \end{equation}
    it is straightforward to conclude that $z_2$ has the properties we stated in the present Lemma.
\end{proof}

\begin{remark}
    Any function $z_2$ complying with the properties we stated Lemma \ref{lem:z2} must differ from the one we exposed in \eqref{eq:particular_z2} by another one which only depends on time. In this connection, the way we chose the boundary layer term $c$ in the definition of $z_1$ was key in the last proof. If we were to seek further higher-order approximation results, we would need to take this time-dependent difference in an appropriate way (thus identifying additional boundary layer terms). For our current purpose, which is to investigate the accuracy of the first-order correction, taking $z_2$ according to \eqref{eq:particular_z2} will suffice.   
\end{remark}

\begin{theorem} \label{thm:AsymptoticnessOrd1}
Let us assume in \eqref{eq:AsymptoticAvg} that $g \in \mathcal{D}\left(\mathcal{L}\right)$ with at most polynomial growth and $\left\langle g \right\rangle = 0$ implies\footnote{Cf. the developments within Section 3.2 from \cite[Eq. (3.10)]{fouque2011multiscale} onwards.}
\begin{equation} \label{eq:polynomial_condn_asymptotic_avg}
    \left| \mathbb{E}_{t,\,\boldsymbol{y}}\left[ g(\boldsymbol{y}^{\epsilon}_T) \right] \right| \leqslant C\left(1+\left|\boldsymbol{y}\right|^n \right) e^{-\frac{a}{\epsilon}(T-t)},
\end{equation}
where $a > 0$ is the spectral gap of $\mathcal{L},$ the constant $C=C(g)$ and the positive integer $n = n(g)$ are sufficiently large, but independent of $t \in \left[0,\,T\right],\, \boldsymbol{y} \in \mathbb{R}^d,$ and $0 < \epsilon \leqslant \epsilon_0,$ for some small enough $\epsilon_0 > 0.$ Then, for $p \geqslant 1$ and $0 \leqslant t_0 < T,$ we have the asymptotics
\begin{equation} \label{eq:accuracy_ord1_uniftime}
    \sup_{t_0\leqslant t \leqslant T }\left\{ (T-t)\mathbb{E}_{t_0,\,\boldsymbol{y}}\left[ \left| z(t,\boldsymbol{y}^\epsilon_t) - \overline{z}_1(t,\boldsymbol{y}^\epsilon_t) \right|^p \right]^{1/p} \right\} = O_p(\epsilon^2), \text{ as } \epsilon \downarrow 0,
\end{equation}
and
\begin{equation} \label{eq:accuracy_ord1_nonblowup}
    \int_{t_0}^T \mathbb{E}_{t_0,\,\boldsymbol{y}}\left[ \left| z(t,\boldsymbol{y}^\epsilon_t) - \overline{z}_1(t,\boldsymbol{y}^\epsilon_t) \right|^p \right]^{1/p} \, dt = O_p(\epsilon^2), \text{ as } \epsilon \downarrow 0.
\end{equation}
\end{theorem}
\begin{proof}
Let us fix $z_2$ according to the construction \eqref{eq:particular_z2} we made in the proof of Lemma \ref{lem:z2}. We deduce from Eqs. \eqref{eq:MainPDE}, \eqref{eq:ODE_OrdZeroApprox}, \eqref{eq:PDE_OrdOneApprox} and \eqref{ThirdEqn} that $R_1 := z - \overline{z}_1 - \epsilon^2 z_2$ must solve
\begin{equation} \label{eq:PDEforR1}
    \partial_t R_1 + \frac{1}{\epsilon} \mathcal{L} R_1 + \phi\kappa^{-1/\phi}\left[|z|^{1+1/\phi} - |z_0|^{1+1/\phi} - \epsilon\left(1+\frac{1}{\phi}\right)|z_0|^{1/\phi}\sgn\left(z_0\right)z_1 \right] + \epsilon^2 \partial_t z_2 = 0.
\end{equation}
Since $z - ( z_0 + \epsilon z_1 ) = R_1 + \epsilon^2z_2,$ we observe that
\begin{align} \label{eq:EqualitiesToDerivePDEofR1}
    \begin{split}
        |z&|^{1+1/\phi} -|z_0|^{1+1/\phi} - \epsilon\left(1+\frac{1}{\phi}\right)|z_0|^{1/\phi}\sgn\left(z_0\right)z_1 \\
        &= \left( |z|^{1+1/\phi}- \left|z_0 + \epsilon z_1\right|^{1+1/\phi} \right) + \left[ \left|z_0 + \epsilon z_1\right|^{1+1/\phi} - |z_0|^{1+1/\phi} - \epsilon\left(1+\frac{1}{\phi}\right)|z_0|^{1/\phi}\sgn\left(z_0\right)z_1 \right] \\
        &= G(z,z_0 + \epsilon z_1)(R_1 +\epsilon^2 z_2) + H(z_0+\epsilon z_1,z_0)(z_1)^2 \epsilon^2 
    \end{split}
\end{align}
where we define
\begin{equation} \label{eq:DefnOfH}
    H\left(s,s^\prime \right) :=\begin{cases}
    \frac{|s|^{1+1/\phi} - |s^\prime|^{1+1/\phi} - (1+1/\phi)|s^\prime|^{1/\phi}\sgn(s^\prime)s}{(s-s^\prime)^2} &\text{ if } s\neq s^\prime,\\
    \frac{1}{2\phi}\left(1 + \frac{1}{\phi}\right)|s|^{1/\phi - 1} &\text{ otherwise.}
    \end{cases}
\end{equation}
Putting \eqref{eq:PDEforR1} and \eqref{eq:EqualitiesToDerivePDEofR1} together, we obtain
$$
\partial_t R_1 + \frac{1}{\epsilon} \mathcal{L} R_1 + g_\epsilon R_1 + \epsilon^2 h_\epsilon = 0,
$$
where we define $g_\epsilon$ and $h_\epsilon$ as 
$$
g_\epsilon\left(t,\boldsymbol{y}\right) := \phi\kappa^{-1/\phi}\left( \boldsymbol{y} \right) G\left(z\left(t,\boldsymbol{y}\right),\,z_0(t) + \epsilon z_1\left(t,\boldsymbol{y}\right) \right),
$$
and
$$
h_\epsilon\left( t,\boldsymbol{y} \right) := \partial_t z_2\left( t,\boldsymbol{y} \right) +G\left(z\left( t,\boldsymbol{y} \right),\,z_0(t) + \epsilon z_1\left( t,\boldsymbol{y} \right) \right)z_2\left( t,\boldsymbol{y} \right) + H\left(z_0(t) + \epsilon z_1\left( t,\boldsymbol{y} \right),\,z_0(t)\right)\left[ z_1\left( t,\boldsymbol{y} \right) \right]^2.
$$

We now proceed to show that $g_\epsilon$ and $h_\epsilon$ are at most polynomially growing in the spatial variable, uniformly in time and also for sufficiently small values of $\epsilon.$ We emphasize that in subsequent estimates, we assume the generic constant $C$ and the integers $n$ to be sufficiently large, and they may change from line to line throughout estimates. However, both $C$ and $n$ are independent of a particular point $(t,\,\boldsymbol{y})$ or of a particular positive $\epsilon\leqslant 1.$

\textbf{CLAIM 1:} \textit{We have $\left| g_\epsilon\left(t,\,\boldsymbol{y}\right) \right| \leqslant C (1+|\boldsymbol{y}|^n),$ for some large enough integer $n,$ where we can take $C > 0$ uniformly with respect to $(t,\,\epsilon) \in \left[0,\,T\right]\times \left]0,\,1\right].$}

Indeed, by using (\textbf{H2}) and \eqref{eq:basic_G_bound}, we can argue as in \eqref{eq:bdd_of_coef_a_leading_ord_approx} to obtain
$$
\left|g_\epsilon(t,\,\boldsymbol{y})\right| \leqslant C \left( 1 + |\boldsymbol{y}|^n + \epsilon |z_1(t,\,\boldsymbol{y})|^{1/\phi} \right) \leqslant C \left( 1+|\boldsymbol{y}|^n\right),
$$
for some sufficiently large positive integer $n,$ as $\epsilon \leqslant 1$ and $z_1$ is at most polynomially growing in $\boldsymbol{y}$ uniformly in time, see \eqref{eq:bddnss_z1} and \eqref{eq:bddns_src_f_leading_order_approx}. 

\textbf{CLAIM 2:} \textit{We have $\left| h_\epsilon\left(t,\,\boldsymbol{y}\right) \right| \leqslant C (1+|\boldsymbol{y}|^n),$ for some large enough integer $n,$ where we can take $C > 0$ uniformly with respect to $(t,\,\epsilon) \in \left[0,\,T\right]\times \left]0,\,\infty\right[.$}

To prove the bound on $h_\epsilon,$ let us first notice that by \eqref{eq:DefnOfH}, we have
$$
|H(s,\,s^\prime)| \leqslant C \left( |s|^{1/\phi - 1} + |s^\prime|^{1/\phi - 1} \right) \hspace{1.0cm} (s,s^\prime \in \mathbb{R}).
$$
Let us recall from (\textbf{H2}) that $0 < \phi \leqslant 1.$ Thus, by invoking once again \eqref{eq:basic_G_bound}, we deduce
\begin{align} \label{eq:bounding_heps_1}
  \begin{split}
    |h_\epsilon\left(t,\,\boldsymbol{y}\right)| &\leqslant C\Big[ \left| \partial_t z_2(t,\,\boldsymbol{y}) \right| + (1+|\boldsymbol{y}|^n)|z_2(t,\,\boldsymbol{y})|  \\
    &\hspace{1.2cm}+ \left( \left| z_0(t) + \epsilon z_1\left( t,\boldsymbol{y} \right)\right|^{1/\phi-1} + \left| z_0(t) \right|^{1/\phi - 1} \right)\left(1+|\boldsymbol{y}|^n\right) \Big] \\
    &\leqslant C\left(1+|\boldsymbol{y}|^n\right)\left( \left| \partial_t z_2(t,\,\boldsymbol{y}) \right| + |z_2(t,\,\boldsymbol{y})| + 1 \right),
  \end{split}
\end{align}
where we used that $1/\phi - 1 \geqslant 0,\, \epsilon \leqslant 1,$ as well as the same bounds on $z_0$ and $z_1$  which we used in the proof of Theorem \ref{thm:AsymptoticnessOrd0}.   By the boundedness of $z_0,$ it is straightforward to conclude that $b_2$ figuring in \eqref{eq:aux_b2_defn} is bounded uniformly in $(t,\,\epsilon).$ Furthermore, since $0< \phi \leqslant 1$, and recalling \eqref{eq:bddnss_z0prime}, we deduce this same boundedness property for $b_2^\prime.$ Hence, using the fact that $\varphi_0,\,\varphi_2,\,\varphi_3,$ and $\varphi_4$ are at most polynomially growing in $\boldsymbol{y},$ we infer from \eqref{eq:particular_z2} the following relation:
\begin{equation} \label{eq:z2_poly}
    \left| \partial_t z_2(t,\,\boldsymbol{y})\right| + \left| z_2(t,\,\boldsymbol{y})\right| \leqslant C\left(1+\left|\boldsymbol{y}\right|^n \right).
\end{equation}
Putting this estimate together with \eqref{eq:bounding_heps_1}, gives the result we stated as \textbf{CLAIM 2}.

Next, we remark that the terminal condition for $R_1$ is
$$
R_1\left(T,\boldsymbol{y}\right) = -\epsilon z_1\left(T,\boldsymbol{y}\right) - \epsilon^2 z_2\left(T,\boldsymbol{y}\right).
$$
Since $z_1(T,\cdot)$ is not identically vanishing but only cancels on a suitable average, we will interpret \eqref{eq:PDEforR1} in the form
$$
\partial_t R_1 + \frac{1}{\epsilon} \mathcal{L} R_1 + v = 0,
$$
for $v := g_\epsilon R_1 +\epsilon^2 h_\epsilon,$ wherefrom the Feynman-Kac Theorem gives the representation
\begin{align} \label{eq:KeyRepresentationOfR1}
    \begin{split}
        R_1\left( t,\boldsymbol{y} \right) &= \mathbb{E}_{t,\boldsymbol{y}}\left[ -\epsilon z_1\left(T,\boldsymbol{y}_T\right) - \epsilon^2 z_2\left(T,\boldsymbol{y}_T\right) + \int_t^T \left\{ g_\epsilon\left(u,\boldsymbol{y}_u\right) R_1\left(u,\boldsymbol{y}_u\right) +\epsilon^2 h_\epsilon\left(u,\boldsymbol{y}_u\right)\right\}\,du \right].
    \end{split}
\end{align}
Above, we have written $\left\{ \boldsymbol{y}_\tau \equiv \boldsymbol{y}^\epsilon_\tau \right\}_\tau,$ omitting the superscript $\epsilon,$ and will continue to do so until the end of the proof. 

By the centering of $z_1(T,\cdot),$ i.e., $\left\langle z_1(T,\cdot) \right\rangle = 0,$ alongside the fact that it belongs to $\mathcal{D}\left(\mathcal{L}\right)$ and that it has at most polynomial growth, we can apply our assumption \eqref{eq:polynomial_condn_asymptotic_avg} to obtain 
\begin{equation} \label{eq:key_bound_on_z1T}
    \left| \mathbb{E}_{t,\boldsymbol{y}}\left[ z_1\left(T,\boldsymbol{y}_T\right) \right] \right| \leqslant C\left(1+\left|\boldsymbol{y}\right|^n \right) e^{-\frac{a}{\epsilon}(T-t)},
\end{equation}
for the spectral gap $a > 0$ of $\mathcal{L},$ cf. Remarks \ref{rem:expression_eigen} and \ref{rem:expectation_centered_bdd_eps}.   

For $p\geqslant 1$ and $0 \leqslant t_0 < T,$ we write
$$
r_{p,\,t_0}(t) := \mathbb{E}_{t_0,\boldsymbol{y}}\left[ \left| R_1\left(t,\boldsymbol{y}_t\right) \right|^p \right]^{1/p}.
$$
In the subsequent estimates, the generic constant $C$ may depend on $p$ and $\boldsymbol{y}$ (through time-uniform bounds on the moments of the path $\left\{ \boldsymbol{y}_t \right\},\,\boldsymbol{y}_0=\boldsymbol{y},$ cf. (\textbf{H4})). 

Now, from \eqref{eq:KeyRepresentationOfR1}, we use the the Claims 1 and 2 we proved above, Eqs. \eqref{eq:z2_poly} and \eqref{eq:key_bound_on_z1T},  the Tower Property of conditional expectations, the H\"older Inequality (writing $p^\prime = p/(p-1)$), and Fubini's Theorem --- whenever it is convenient --- as well as (\textbf{H4}), to derive for $t \in \left[t_0,T \right[$:
\begin{align*}
    r_{p,\,t_0}(t)^p &\leqslant\, C\epsilon^p \mathbb{E}_{t_0,\,\boldsymbol{y}}\left[\left| \mathbb{E}_{t,\boldsymbol{y}_t}\left[ z_1\left(T,\boldsymbol{y}_T\right) \right] \right|^p\right]+ C\epsilon^{2p} \mathbb{E}_{t_0,\,\boldsymbol{y}}\left[ \left| z_2\left(T,\boldsymbol{y}_T\right) \right|^p\right] \\
    &\hspace{0.5cm}+ C\left\{ \int_t^T \mathbb{E}_{t_0,\,\boldsymbol{y}} \left[ | g_\epsilon(u,\,\boldsymbol{y}_u) |^{p^\prime} \,du \right]\right\}^{p/p^\prime}  \int_t^T \mathbb{E}_{t_0,\,\boldsymbol{y}} \left[ \left| R_1(u,\,\boldsymbol{y}_u) \right|^p \right]\, du  \\
    &\hspace{0.5cm}+C\epsilon^{2p}\mathbb{E}_{t_0,\,\boldsymbol{y}}\left[ \int_t^T | h_\epsilon(u,\,\boldsymbol{y}_u) |^p \, du  \right]  \\
    &\leqslant C\epsilon^p e^{-\frac{p\,a}{\epsilon}(T-t)}\left( 1+\sup_{\tau \geqslant 0} \mathbb{E}_{0,\,\boldsymbol{y}}\left[\left| \boldsymbol{y}^1_\tau \right|^n \right] \right) + C\epsilon^{2p}\left( 1+\sup_{\tau \geqslant 0} \mathbb{E}_{0,\,\boldsymbol{y}}\left[\left| \boldsymbol{y}^1_\tau \right|^n \right] \right) \\
    &\hspace{0.5cm}+ C \left( 1+ \sup_{\tau \geqslant 0} \mathbb{E}_{0,\,\boldsymbol{y}} \left[\left| \boldsymbol{y}^1_\tau \right|^n \right] \right)\int_t^T r_{p,\,t_0}(u)^p\,du + C\epsilon^{2p}\left( 1+ \sup_{\tau \geqslant 0} \mathbb{E}_{0,\,\boldsymbol{y}} \left[ \left| \boldsymbol{y}^1_\tau \right|^n\right] \right) \\
    &\leqslant C\epsilon^p e^{-\frac{p\,a}{\epsilon}(T-t)} + C\epsilon^{2p} + C\int_t^T r_{p,\,t_0}(u)^p\,du.
\end{align*}
From these estimates, we employ Gronwall's Lemma and take $1/p-$th powers --- using their monotonicity and subadditivity --- to deduce
\begin{equation} \label{eq:final_estimate_r}
    r_{p,\,t_0}(t) \leqslant C\epsilon e^{-\frac{a}{\epsilon}(T-t)} + C\epsilon^{2}.
\end{equation}
By multiplying \eqref{eq:final_estimate_r} by $(T-t)$ and using that 
$$
(T-t)e^{-\frac{a}{\epsilon}(T-t)} \leqslant \frac{\epsilon}{a},
$$
we obtain \eqref{eq:accuracy_ord1_uniftime}. Alternatively, by integrating \eqref{eq:final_estimate_r} with respect to time, we get \eqref{eq:accuracy_ord1_nonblowup}. This finishes the proof of the Theorem.
\end{proof}

Consequently, we get from Theorem \ref{thm:AsymptoticnessOrd1} (specifically, from \eqref{eq:accuracy_ord1_uniftime})  the following pointwise result away from the boundary.

\begin{corollary}
    Given $\delta > 0$ and $(t,\,\boldsymbol{y}) \in \left[0,T-\delta\right]\times \mathbb{R}^d,$ we have
    $$
    z(t,\,\boldsymbol{y}) = \overline{z}_1(t,\,\boldsymbol{y}) + O_{\delta}\left(\epsilon^2\right),\, \text{ as } \epsilon \downarrow 0.
    $$
\end{corollary}

\section{Conclusions} \label{sec:conclusions}

We investigated the optimal execution problem under a framework where both liquidity and volatility are stochastic. We modeled the uncertainty in the aspects by stipulating that a multi-dimensional Markovian factor drives them. Building upon previous results on this setting, we further proposed to assume that this stochastic driver is a fast mean-reverting process. We used the resulting ergodicity of the system to obtain approximations to the optimal strategy, which yielded insightful strategies that are suited to the types of markets we studied --- particularly, for highly volatile financial products.

Firstly, we derived the leading-order approximation. We then assessed the strategy we obtained by carrying out $10000$ Monte Carlo simulations. We observed the desired effect in the strategy's behavior relative to the risk aversion parameter. It indicated that some positive risk aversion level might be beneficial to the trader, even from the pure viewpoint of wealth maximization. Subsequently, we derived the first-order correction to the leading-order approximation. It included a boundary layer term due to the problem's singularity at the terminal time (mitigated by assuming a finite terminal inventory penalization, which regularizes the problem). We computed the first-order correction via an iterative numerical scheme, and we have seen that it is indeed a consistent --- although marginal --- improvement of the leading-order approximation. 

Finally, we concluded the work by establishing some accuracy results concerning our approximations. We used appropriate representations of the solutions of the HJB via the Feynman-Kac Theorem in such a way that we were able to derive conclusions regarding the asymptotics of our approximations. For the leading-order one, we provided pointwise results rather directly. For the first-order correction, the analysis was more subtle. In this instance, we were able to identify suitable topologies for which the desired accuracy results hold, from where we obtained a similar pointwise consequence.

\section*{Acknowlegements}

YT was financed in part by Coordena\c{c}\~ao de Aperfei\c{c}oamento de Pessoal de N\'ivel Superior - Brasil (CAPES) - Finance code 001.

%-------------------------------------------------------------------------------
% REFERENCES
%-------------------------------------------------------------------------------

%\bibliographystyle{elsarticle-num}
\bibliographystyle{apalike}
\bibliography{References}

\begin{thebibliography}{}

\bibitem[Almgren, 2012]{almgren2012optimal}
Almgren, R. (2012).
\newblock Optimal trading with stochastic liquidity and volatility.
\newblock {\em SIAM Journal on Financial Mathematics}, 3(1):163--181.

\bibitem[Almgren and Chriss, 2001]{almgren2001optimal}
Almgren, R. and Chriss, N. (2001).
\newblock Optimal execution of portfolio transactions.
\newblock {\em Journal of Risk}, 3:5--40.

\bibitem[Almgren and Li, 2016]{almgren2016option}
Almgren, R. and Li, T.~M. (2016).
\newblock Option hedging with smooth market impact.
\newblock {\em Market microstructure and liquidity}, 2(01):1650002.

\bibitem[Almgren et~al., 2005]{almgren2005direct}
Almgren, R., Thum, C., Hauptmann, E., and Li, H. (2005).
\newblock Direct estimation of equity market impact.
\newblock {\em Risk}, 18(7):58--62.

\bibitem[Almgren, 2003]{almgren2003optimal}
Almgren, R.~F. (2003).
\newblock Optimal execution with nonlinear impact functions and
  trading-enhanced risk.
\newblock {\em Applied mathematical finance}, 10(1):1--18.

\bibitem[Alvarez and Bardi, 2002]{alvarez2002viscosity}
Alvarez, O. and Bardi, M. (2002).
\newblock Viscosity solutions methods for singular perturbations in
  deterministic and stochastic control.
\newblock {\em SIAM journal on control and optimization}, 40(4):1159--1188.

\bibitem[Alvarez and Bardi, 2003]{alvarez2003singular}
Alvarez, O. and Bardi, M. (2003).
\newblock Singular perturbations of nonlinear degenerate parabolic pdes: a
  general convergence result.
\newblock {\em Archive for rational mechanics and analysis}, 170:17--61.

\bibitem[Alvarez and Bardi, 2010]{alvarez2010ergodicity}
Alvarez, O. and Bardi, M. (2010).
\newblock {\em Ergodicity, stabilization, and singular perturbations for
  Bellman-Isaacs equations}.
\newblock American Mathematical Soc.

\bibitem[Back and Baruch, 2004]{back2004information}
Back, K. and Baruch, S. (2004).
\newblock Information in securities markets: Kyle meets glosten and milgrom.
\newblock {\em Econometrica}, 72(2):433--465.

\bibitem[Bank et~al., 2017]{bank2017hedging}
Bank, P., Soner, H.~M., and Vo{\ss}, M. (2017).
\newblock Hedging with temporary price impact.
\newblock {\em Mathematics and financial economics}, 11(2):215--239.

\bibitem[Bardi and Cesaroni, 2011]{bardi2011optimal}
Bardi, M. and Cesaroni, A. (2011).
\newblock Optimal control with random parameters: a multiscale approach.
\newblock {\em European journal of control}, 17(1):30--45.

\bibitem[Bardi et~al., 2014]{bardi2014large}
Bardi, M., Cesaroni, A., and Ghilli, D. (2014).
\newblock Large deviations for some fast stochastic volatility models by
  viscosity methods.
\newblock {\em arXiv preprint arXiv:1405.3206}.

\bibitem[Bardi et~al., 2010]{bardi2010convergence}
Bardi, M., Cesaroni, A., and Manca, L. (2010).
\newblock Convergence by viscosity methods in multiscale financial models with
  stochastic volatility.
\newblock {\em SIAM Journal on Financial Mathematics}, 1(1):230--265.

\bibitem[Bardi and Kouhkouh, 2022]{bardi2022deep}
Bardi, M. and Kouhkouh, H. (2022).
\newblock Deep relaxation of controlled stochastic gradient descent via
  singular perturbations.
\newblock {\em arXiv preprint arXiv:2209.05564}.

\bibitem[Bardi and Kouhkouh, 2023]{bardi2023singular}
Bardi, M. and Kouhkouh, H. (2023).
\newblock Singular perturbations in stochastic optimal control with unbounded
  data.
\newblock {\em ESAIM: Control, Optimisation and Calculus of Variations}, 29:52.

\bibitem[Bayraktar and Ludkovski, 2011]{bayraktar2011optimal}
Bayraktar, E. and Ludkovski, M. (2011).
\newblock Optimal trade execution in illiquid markets.
\newblock {\em Mathematical Finance: An International Journal of Mathematics,
  Statistics and Financial Economics}, 21(4):681--701.

\bibitem[Becherer et~al., 2018]{becherer2018optimal}
Becherer, D., Bilarev, T., and Frentrup, P. (2018).
\newblock Optimal liquidation under stochastic liquidity.
\newblock {\em Finance and Stochastics}, 22(1):39--68.

\bibitem[Beckner, 1989]{beckner1989generalized}
Beckner, W. (1989).
\newblock A generalized poincar{\'e} inequality for gaussian measures.
\newblock {\em Proceedings of the American Mathematical Society}, pages
  397--400.

\bibitem[Bertsimas and Lo, 1998]{bertsimas1998optimal}
Bertsimas, D. and Lo, A.~W. (1998).
\newblock Optimal control of execution costs.
\newblock {\em Journal of Financial Markets}, 1(1):1--50.

\bibitem[Bouchaud, 2010]{bouchaud2010price}
Bouchaud, J.-P. (2010).
\newblock Price impact.
\newblock {\em Encyclopedia of quantitative finance}.

\bibitem[Breiman, 1996]{breiman1996bagging}
Breiman, L. (1996).
\newblock Bagging predictors.
\newblock {\em Machine learning}, 24:123--140.

\bibitem[Cartea and Jaimungal, 2015]{cartea2015optimal}
Cartea, A. and Jaimungal, S. (2015).
\newblock Optimal execution with limit and market orders.
\newblock {\em Quantitative Finance}, 15(8):1279--1291.

\bibitem[Cartea and Jaimungal, 2016]{cartea2016incorporating}
Cartea, A. and Jaimungal, S. (2016).
\newblock Incorporating order-flow into optimal execution.
\newblock {\em Mathematics and Financial Economics}, 10(3):339--364.

\bibitem[Cartea et~al., 2015]{cartea2015algorithmic}
Cartea, {\'A}., Jaimungal, S., and Penalva, J. (2015).
\newblock {\em Algorithmic and high-frequency trading}.
\newblock Cambridge University Press.

\bibitem[Cheridito and Sepin, 2014]{cheridito2014optimal}
Cheridito, P. and Sepin, T. (2014).
\newblock Optimal trade execution under stochastic volatility and liquidity.
\newblock {\em Applied Mathematical Finance}, 21(4):342--362.

\bibitem[Cont et~al., 2014]{cont2014price}
Cont, R., Kukanov, A., and Stoikov, S. (2014).
\newblock The price impact of order book events.
\newblock {\em Journal of financial econometrics}, 12(1):47--88.

\bibitem[Evangelista et~al., 2022]{evangelista2022price}
Evangelista, D., Saporito, Y., and Thamsten, Y. (2022).
\newblock Price formation in financial markets: a game-theoretic perspective.
\newblock {\em arXiv preprint arXiv:2202.11416}.

\bibitem[Evangelista and Thamsten, 2020]{evangelista2020finite}
Evangelista, D. and Thamsten, Y. (2020).
\newblock On finite population games of optimal trading.
\newblock {\em arXiv preprint arXiv:2004.00790}.

\bibitem[F{\'e}ron et~al., 2020]{feron2020price}
F{\'e}ron, O., Tankov, P., and Tinsi, L. (2020).
\newblock Price formation and optimal trading in intraday electricity markets
  with a major player.
\newblock {\em Risks}, 8(4):133.

\bibitem[F{\'e}ron et~al., 2021]{feron2021price}
F{\'e}ron, O., Tankov, P., and Tinsi, L. (2021).
\newblock Price formation and optimal trading in intraday electricity markets.
\newblock In {\em Network Games, Control and Optimization: 10th International
  Conference, NetGCooP 2020, France, September 22--24, 2021, Proceedings 10},
  pages 294--305. Springer.

\bibitem[Fouque et~al., 2021]{fouque2021optimal}
Fouque, J.-P., Jaimungal, S., and Saporito, Y.~F. (2021).
\newblock Optimal trading with signals and stochastic price impact.
\newblock {\em arXiv preprint arXiv:2101.10053}.

\bibitem[Fouque et~al., 2011]{fouque2011multiscale}
Fouque, J.-P., Papanicolaou, G., Sircar, R., and S{\o}lna, K. (2011).
\newblock {\em Multiscale stochastic volatility for equity, interest rate, and
  credit derivatives}.
\newblock Cambridge University Press.

\bibitem[Fruth et~al., 2019]{fruth2019optimal}
Fruth, A., Sch{\"o}neborn, T., and Urusov, M. (2019).
\newblock Optimal trade execution in order books with stochastic liquidity.
\newblock {\em Mathematical Finance}, 29(2):507--541.

\bibitem[Fu et~al., 2022]{fu2022portfolio}
Fu, G., Horst, U., and Xia, X. (2022).
\newblock Portfolio liquidation games with self-exciting order flow.
\newblock {\em Mathematical Finance}, 32(4):1020--1065.

\bibitem[Fujii, 2019]{fujii2019probabilistic}
Fujii, M. (2019).
\newblock Probabilistic approach to mean field games and mean field type
  control problems with multiple populations.
\newblock {\em arXiv preprint arXiv:1911.11501}.

\bibitem[Fujii, 2022]{fujii2022equilibrium}
Fujii, M. (2022).
\newblock Equilibrium pricing of securities in the co-presence of cooperative
  and non-cooperative populations.
\newblock {\em arXiv preprint arXiv:2209.12639}.

\bibitem[Fujii and Sekine, 2023]{fujii2023mean}
Fujii, M. and Sekine, M. (2023).
\newblock Mean-field equilibrium price formation with exponential utility.
\newblock {\em arXiv preprint arXiv:2304.07108}.

\bibitem[Fujii and Takahashi, 2021]{fujii2021equilibrium}
Fujii, M. and Takahashi, A. (2021).
\newblock Equilibrium price formation with a major player and its mean field
  limit.
\newblock {\em arXiv preprint arXiv:2102.10756}.

\bibitem[Fujii and Takahashi, 2022a]{fujii2022mean}
Fujii, M. and Takahashi, A. (2022a).
\newblock A mean field game approach to equilibrium pricing with market
  clearing condition.
\newblock {\em SIAM Journal on Control and Optimization}, 60(1):259--279.

\bibitem[Fujii and Takahashi, 2022b]{fujii2022strong}
Fujii, M. and Takahashi, A. (2022b).
\newblock Strong convergence to the mean field limit of a finite agent
  equilibrium.
\newblock {\em SIAM Journal on Financial Mathematics}, 13(2):459--490.

\bibitem[Gatheral and Oomen, 2010]{gatheral2010zero}
Gatheral, J. and Oomen, R.~C. (2010).
\newblock Zero-intelligence realized variance estimation.
\newblock {\em Finance and Stochastics}, 14(2):249--283.

\bibitem[Gatheral and Schied, 2011]{gatheral2011optimal}
Gatheral, J. and Schied, A. (2011).
\newblock Optimal trade execution under geometric brownian motion in the
  almgren and chriss framework.
\newblock {\em International Journal of Theoretical and Applied Finance},
  14(03):353--368.

\bibitem[Gatheral et~al., 2012]{gatheral2012transient}
Gatheral, J., Schied, A., and Slynko, A. (2012).
\newblock Transient linear price impact and fredholm integral equations.
\newblock {\em Mathematical Finance: An International Journal of Mathematics,
  Statistics and Financial Economics}, 22(3):445--474.

\bibitem[Glosten and Milgrom, 1985]{glosten1985bid}
Glosten, L.~R. and Milgrom, P.~R. (1985).
\newblock Bid, ask and transaction prices in a specialist market with
  heterogeneously informed traders.
\newblock {\em Journal of financial economics}, 14(1):71--100.

\bibitem[Graewe and Horst, 2017]{graewe2017optimal}
Graewe, P. and Horst, U. (2017).
\newblock Optimal trade execution with instantaneous price impact and
  stochastic resilience.
\newblock {\em SIAM Journal on Control and Optimization}, 55(6):3707--3725.

\bibitem[Graewe et~al., 2018]{graewe2018smooth}
Graewe, P., Horst, U., and S{\'e}r{\'e}, E. (2018).
\newblock Smooth solutions to portfolio liquidation problems under
  price-sensitive market impact.
\newblock {\em Stochastic Processes and their Applications}, 128(3):979--1006.

\bibitem[Gu{\'e}ant, 2016]{gueant2016financial}
Gu{\'e}ant, O. (2016).
\newblock {\em The Financial Mathematics of Market Liquidity: From optimal
  execution to market making}, volume~33.
\newblock CRC Press.

\bibitem[Gu{\'e}ant et~al., 2012]{gueant2012optimal}
Gu{\'e}ant, O., Lehalle, C.-A., and Fernandez-Tapia, J. (2012).
\newblock Optimal portfolio liquidation with limit orders.
\newblock {\em SIAM Journal on Financial Mathematics}, 3(1):740--764.

\bibitem[Gu{\'e}ant and Pu, 2017]{gueant2017option}
Gu{\'e}ant, O. and Pu, J. (2017).
\newblock Option pricing and hedging with execution costs and market impact.
\newblock {\em Mathematical Finance}, 27(3):803--831.

\bibitem[Hol{\`y} and Tomanov{\'a}, 2018]{holy2018estimation}
Hol{\`y}, V. and Tomanov{\'a}, P. (2018).
\newblock Estimation of ornstein-uhlenbeck process using ultra-high-frequency
  data with application to intraday pairs trading strategy.
\newblock {\em arXiv preprint arXiv:1811.09312}.

\bibitem[Horst and Xia, 2020]{horst2020continuous}
Horst, U. and Xia, X. (2020).
\newblock Continuous viscosity solutions to linear-quadratic stochastic control
  problems with singular terminal state constraint.
\newblock {\em Applied Mathematics \& Optimization}, pages 1--26.

\bibitem[Jaimungal et~al., 2023]{jaimungal2023optimal}
Jaimungal, S., Saporito, Y.~F., Souza, M.~O., and Thamsten, Y. (2023).
\newblock Optimal trading in automatic market makers with deep learning.
\newblock {\em arXiv preprint arXiv:2304.02180}.

\bibitem[Konishi, 2002]{konishi2002optimal}
Konishi, H. (2002).
\newblock Optimal slice of a vwap trade.
\newblock {\em Journal of Financial Markets}, 5(2):197--221.

\bibitem[Krishnan, 1992]{krishnan1992equivalence}
Krishnan, M. (1992).
\newblock An equivalence between the kyle (1985) and the glosten—milgrom
  (1985) models.
\newblock {\em Economics Letters}, 40(3):333--338.

\bibitem[Kyle, 1985]{kyle1985continuous}
Kyle, A.~S. (1985).
\newblock Continuous auctions and insider trading.
\newblock {\em Econometrica: Journal of the Econometric Society}, pages
  1315--1335.

\bibitem[Laruelle and Lehalle, 2018]{laruelle2018market}
Laruelle, S. and Lehalle, C.-a. (2018).
\newblock {\em Market microstructure in practice}.
\newblock World Scientific.

\bibitem[Lillo et~al., 2003]{lillo2003master}
Lillo, F., Farmer, J.~D., and Mantegna, R.~N. (2003).
\newblock Master curve for price-impact function.
\newblock {\em Nature}, 421(6919):129--130.

\bibitem[Loeb, 1983]{loeb1983trading}
Loeb, T.~F. (1983).
\newblock Trading cost: the critical link between investment information and
  results.
\newblock {\em Financial Analysts Journal}, 39(3):39--44.

\bibitem[Micheli et~al., 2021]{micheli2021closed}
Micheli, A., Muhle-Karbe, J., and Neuman, E. (2021).
\newblock Closed-loop nash competition for liquidity.
\newblock {\em arXiv preprint arXiv:2112.02961}.

\bibitem[Moazeni et~al., 2013]{moazeni2013optimal}
Moazeni, S., Coleman, T.~F., and Li, Y. (2013).
\newblock Optimal execution under jump models for uncertain price impact.
\newblock {\em Journal of Computational Finance}, 16(4):1--44.

\bibitem[Neuman and Vo{\ss}, 2023]{neuman2023trading}
Neuman, E. and Vo{\ss}, M. (2023).
\newblock Trading with the crowd.
\newblock {\em Mathematical Finance}.

\bibitem[Obizhaeva and Wang, 2013]{obizhaeva2013optimal}
Obizhaeva, A.~A. and Wang, J. (2013).
\newblock Optimal trading strategy and supply/demand dynamics.
\newblock {\em Journal of Financial Markets}, 16(1):1--32.

\bibitem[Siu et~al., 2019]{siu2019optimal}
Siu, C.~C., Guo, I., Zhu, S.-P., and Elliott, R.~J. (2019).
\newblock Optimal execution with regime-switching market resilience.
\newblock {\em Journal of Economic Dynamics and Control}, 101:17--40.

\bibitem[Souza, 2022]{souza2021regularized}
Souza, M.~O. (2022).
\newblock On regularized optimal execution problems and their singular limits.
\newblock {\em Applied Mathematical Finance}, 29(2):79--109.

\bibitem[Walia, 2006]{walia2006optimal}
Walia, N. (2006).
\newblock Optimal trading: Dynamic stock liquidation strategies.
\newblock {\em Senior thesis, Princeton University}.

\bibitem[Zhang et~al., 2005]{zhang2005tale}
Zhang, L., Mykland, P.~A., and A{\"\i}t-Sahalia, Y. (2005).
\newblock A tale of two time scales: Determining integrated volatility with
  noisy high-frequency data.
\newblock {\em Journal of the American Statistical Association},
  100(472):1394--1411.

\end{thebibliography}

\appendix
\section{On the numerical computation of the first-order correction} \label{app:LeadingOrder}

We circumvent the issue that zero is an eigenvalue of $\mathcal{L}$ by taking $\eta > 0,$ in such a way that $\eta$ does not belong to the spectrum of $\mathcal{L}.$ From now on, we suppose that for each sufficiently small such $\eta,$ the operator $\mathcal{L} + \eta I$ is not only injective but that its image contains all continuous functions with at most polynomial growth. We propose the following algorithm to compute the\footnote{The fact that there is at most one follows from (\textbf{H3}) and the zero average condition.} solution $\varphi$ to the PDE
\begin{equation} \label{eq:PoissonAvgZero}
    \begin{cases}
        \mathcal{L}\varphi = f \text{ in } \mathbb{R}^d,\\
        \left\langle \varphi \right\rangle = 0, 
    \end{cases}
\end{equation}
where $f \in C(\mathbb{R}^d)$ is centered and at most polynomially growing:

\begin{algorithm}[H] \label{algo:NumericalAlgo}
\SetAlgoLined
\KwResult{Numerical solution of \eqref{eq:PoissonAvgZero}}
 Initialize $\varphi^0 \equiv 0,$ $k=0,$ the error variable $\epsilon,$ and stipulate the tolerance $\epsilon_0$\;
 \While{$\epsilon \geqslant \epsilon_0$}{
    $1.$ Solve 
    $$
    \mathcal{L}\varphi^{k+1} + \eta \varphi^{k+1} = f + \eta \varphi^k \text{ in } \mathbb{R}^d,
    $$\\
    $2.$ Update $\epsilon;$\\
    $3.$ $k \gets k+1.$
 }
 \Return{$\varphi^k$}
 \caption{Iterative numerical algorithm for solving the PDE \eqref{eq:PoissonAvgZero}.}
\end{algorithm}

\begin{lemma}\label{lem:Poincare}
Let us assume that $d=2,$ that the parameters of the model are as in Subsection \ref{subsec:Numerics1} --- in particular, we remark that $\Pi$ is given by \eqref{eq:Pi2DOUcase}. Then, for every $f \in L^2(\Pi)$ such that $\partial_{\boldsymbol{y}_1}f,\,\partial_{\boldsymbol{y}_2}f \in L^2(\Pi)$ and $\left\langle f \right\rangle = 0,$ we have
$$
\left\langle f^2 \right\rangle \leqslant C \left\langle |\nabla f |^2\right\rangle,
$$
where $C$ is independent of $f.$
\end{lemma}
\begin{proof}
In effect, since the matrix $\boldsymbol{A}$ we defined in Eq. \eqref{eq:MatrixA} is symmetric and positive definite, we can find a symmetric invertible $\boldsymbol{S} \in \mathbb{R}^{2\times 2}$ such that $\boldsymbol{S}^2 = \boldsymbol{A}.$ From Beckner's inequality for Gaussian measures, cf. \cite[Theorem 1, inequality (2)]{beckner1989generalized}, we know that
\begin{equation} \label{eq:BecknerIneq}
    \int_{\mathbb{R}^d} |g(\boldsymbol{x})|^2\frac{e^{-\frac{|\boldsymbol{x}|^2}{2}}}{2\pi}\,d\boldsymbol{x} \leqslant \int_{\mathbb{R}^d} |\nabla g(\boldsymbol{x})|^2\frac{e^{-\frac{|\boldsymbol{x}|^2}{2}}}{2\pi}\,d\boldsymbol{x},
\end{equation}
for every $g,\,\partial_{\boldsymbol{x}_1}g,\,\partial_{\boldsymbol{x}_2}g \in L^2(\Pi)$ and $\int_{\mathbb{R}^2}g(\boldsymbol{x})\frac{e^{-\frac{|\boldsymbol{x}|^2}{2}}}{2\pi}\,d\boldsymbol{x} = 0.$ Given $f$ as in the statement of the current Lemma, let us set $g(\boldsymbol{x}) := f\left( \boldsymbol{S} \boldsymbol{x} + \boldsymbol{m} \right).$ We notice that $g$ meets the conditions for us to apply inequality \eqref{eq:BecknerIneq}; the change of variables $\boldsymbol{y} = \boldsymbol{S}\boldsymbol{x} + \boldsymbol{m}$ yields
\begin{align*}
    \left\langle f^2 \right\rangle &= \int_{\mathbb{R}^d} |g\left(\boldsymbol{S}^{-1}\left(\boldsymbol{y}-\boldsymbol{m}\right)\right)|^2\Pi(\boldsymbol{y})\,\,d\boldsymbol{y} \\
    &= \int_{\mathbb{R}^d} |g(\boldsymbol{x})|^2\frac{e^{-\frac{|\boldsymbol{x}|^2}{2}}}{2\pi}\,d\boldsymbol{x} \\
    &\leqslant \int_{\mathbb{R}^d} |\nabla g(\boldsymbol{x})|^2\frac{e^{-\frac{|\boldsymbol{x}|^2}{2}}}{2\pi}\,d\boldsymbol{x} \\
    &= \int_{\mathbb{R}^d} |\left(\nabla g\right)\left(\boldsymbol{S}^{-1}\left(\boldsymbol{y}-\boldsymbol{m}\right)\right)|^2\Pi(\boldsymbol{y})\,\,d\boldsymbol{y} \\
    &= \int_{\mathbb{R}^d} |\boldsymbol{S}\nabla f\left( \boldsymbol{y} \right)|^2\Pi(\boldsymbol{y})\,\,d\boldsymbol{y} \\
    &\leqslant C\left\langle \left|\nabla f \right|^2 \right\rangle,
\end{align*}
as we wanted to show.
\end{proof}

\begin{proposition} \label{prop:PoissonNumericalSolution}
Let us suppose that the assumptions of Lemma \eqref{lem:Poincare} are in force, and that, for every sufficiently small $\eta > 0$ and $f \in C\left(\mathbb{R}^2\right)$ at most polynomially growing, the equation 
$$
\mathcal{L}\varphi+\eta \varphi = f \text{ in } \mathbb{R}^2,
$$
admits a unique at most polynomially growing solution $\varphi \in \mathcal{D}\left( \mathcal{L} \right)$ whose derivatives up to order two also share this growth property. Let us denote by $\left\{\varphi^k\right\}_k$ the sequence resulting from Algorithm \ref{algo:NumericalAlgo}. Then, the following claims are true:
\begin{itemize}
    \item[(a)] The sequence $\left\{ \varphi^k \right\}_k$ is well-defined;
    \item[(b)] If $f$ is centered, then we have $\left\langle \varphi^k \right\rangle = 0,$ for every positive integer $k;$
    \item[(c)] If we fix $\eta > 0$ sufficiently small, then there exists $\varphi \in L^2(\Pi)$ such that $\varphi^k \xrightarrow{k \rightarrow \infty} \varphi$ in $L^2(\Pi)$ and almost everywhere, as $k \rightarrow \infty.$
\end{itemize}
\end{proposition}
\begin{proof}
Proving $(a)$ by induction is straightforward. To prove $(b),$ it suffices to notice that 
$$
\eta\left\langle \varphi^k \right\rangle = \left\langle f + \eta\varphi^k \right\rangle = \left\langle \mathcal{L}\varphi^{k+1} + \eta \varphi^{k+1}\right\rangle = \eta \left\langle \varphi^{k+1} \right\rangle,
$$
since $\left\langle f \right\rangle = 0.$ Recalling that $\varphi^0 \equiv 0,$ the above relations inductively imply what we stated in $(b).$ Now, let us show $(c).$ We write $e^{k} := \varphi^k - \varphi^{k-1},$ for positive $k.$ It follows that
\begin{equation} \label{eq:Error_k_PDE}
    \mathcal{L} e^{k+1} + \eta e^{k+1} = \eta e^k \text{ in } \mathbb{R}^d.
\end{equation}
We claim that 
\begin{equation} \label{eq:EnergyEstimate_Error_k}
    - \left\langle e^{k+1} \mathcal{L}e^{k+1} \right\rangle \geqslant C \left\langle \left( e^{k+1} \right)^2\right\rangle 
\end{equation}
In effect, carrying out some integration by parts, we have that
\begin{align*}
    - \left\langle e^{k+1} \mathcal{L}e^{k+1} \right\rangle =& \int_{\mathbb{R}^2} \left[ \frac{\sigma^2_1}{2}\left(\partial_{\boldsymbol{y}_1}e^{k+1}\right)^2 + \rho \sigma_1\sigma_2 \partial_{\boldsymbol{y}_1}e^{k+1}\partial_{\boldsymbol{y}_2}e^{k+1} + \frac{\sigma^2_2}{2}\left(\partial_{\boldsymbol{y}_2}e^{k+1}\right)^2 \right]\Pi\,\,d\boldsymbol{y} \\
    &+ \int_{\mathbb{R}^2} \left\{ \frac{\sigma^2_1}{2}\partial_{\boldsymbol{y}_1}\left[ \frac{\left( e^{k+1} \right)^2}{2} \right] \partial_{\boldsymbol{y}_1}\Pi  + \frac{\sigma^2_2}{2}\partial_{\boldsymbol{y}_2}\left[ \frac{\left( e^{k+1} \right)^2}{2} \right] \partial_{\boldsymbol{y}_2}\Pi + \rho \sigma_1\sigma_2\partial_{\boldsymbol{y}_1}\left[\frac{\left( e^{k+1} \right)^2}{2}\right]\partial_{\boldsymbol{y}_2}\Pi \right\}\,\,d\boldsymbol{y}\\
    &+ \int_{\mathbb{R}^2} \left\{ \lambda_1 (m_1-\boldsymbol{y}_1) \partial_{\boldsymbol{y}_1}\left[\frac{\left( e^{k+1} \right)^2}{2}\right] + \lambda_2 (m_2-\boldsymbol{y}_2)\partial_{\boldsymbol{y}_2}\left[\frac{\left( e^{k+1} \right)^2}{2}\right] \right\}\Pi\,\,d\boldsymbol{y} \\
    =& - \int_{\mathbb{R}^2} \frac{\left( e^{k+1} \right)^2}{2} \mathcal{L}^* \Pi \,d\boldsymbol{y} \\
    &+ \int_{\mathbb{R}^2} \left[ \frac{\sigma^2_1}{2}\left( \partial_{\boldsymbol{y}_1}e^{k+1}\right)^2 + \frac{\sigma^2_2}{2}\left( \partial_{\boldsymbol{y}_2}e^{k+1}\right)^2 + \rho \sigma_1\sigma_2 \partial_{\boldsymbol{y}_1}e^{k+1} \partial_{\boldsymbol{y}_2}e^{k+1} \right] \Pi \,\,d\boldsymbol{y} \\
    \geqslant& \frac{1}{2}\min\left(\sigma_1^2,\sigma^2_2\right)(1-\rho)\left\langle \left|\nabla e^{k+1} \right|^2 \right\rangle.
\end{align*}
We emphasize that the boundary terms all vanish due to the fact that $e^{k+1}$ and its derivatives are at most polynomially growing, whereas $\Pi$ is in the Schwartz class. Moreover, from Lemma \ref{lem:Poincare}, we have the Poincar\'e inequality for our Gaussian weight:
\begin{equation} \label{eq:PoincareGaussian}
\left\langle \left|\nabla e^{k+1} \right|^2 \right\rangle \geqslant C \left\langle \left| e^{k+1} \right|^2 \right\rangle
\end{equation}

We multiply Eq. \eqref{eq:Error_k_PDE} by $-e^{k+1}\Pi$ and integrate over $\mathbb{R}^2$ to obtain
\begin{equation} \label{eq:FirstEq_Error_k}
    -\left\langle e^{k+1}\mathcal{L}e^{k+1} \right\rangle - \eta \left\langle \left(e^{k+1}\right)^2\right\rangle = -\eta\left\langle e^k e^{k+1} \right\rangle.
\end{equation}
Using Eq. \eqref{eq:EnergyEstimate_Error_k} in Eq. \eqref{eq:FirstEq_Error_k}, as well as Young's inequality, we derive
$$
\left( C - \eta \right)\left\langle \left(e^{k+1}\right)^2\right\rangle \leqslant -\eta\left\langle e^k e^{k+1} \right\rangle \leqslant \frac{\eta}{2}\left(\left\langle \left(e^{k+1}\right)^2\right\rangle + \left\langle \left(e^{k}\right)^2\right\rangle \right),
$$
whence
\begin{equation} \label{eq:ConcludingEq_Error_k}
    \left\langle \left(e^{k+1}\right)^2\right\rangle \leqslant \frac{\eta}{2C - 3\eta}\left\langle \left(e^{k}\right)^2\right\rangle.
\end{equation}
For sufficiently small $\eta > 0,$ viz., $\eta < C/2,$ Eq. \eqref{eq:ConcludingEq_Error_k} allows us to conclude that 
$$
\sum_k \left\langle \left(e^{k}\right)^2\right\rangle < \infty,
$$
thus finishing the proof of item $(c).$  
\end{proof}

\begin{corollary}
The function $\varphi$ we identified in Propositon \ref{prop:PoissonNumericalSolution} satisfies $\nabla \varphi \in L^2(\Pi),\,\left\langle \varphi \right\rangle =0,$ and solves \eqref{eq:PoissonAvgZero} distributionally.
% , i.e., 
% \begin{equation} \label{eq:WeakSolnPoissonAvgZero}
%     \begin{cases}
%         \int_{\mathbb{R}^2}\left[ \frac{\sigma_1^2}{2}\partial_{\boldsymbol{y}_1}\varphi\partial_{\boldsymbol{y}_1} \psi + \frac{\sigma_2^2}{2}\partial_{\boldsymbol{y}_2}\varphi\partial_{\boldsymbol{y}_2} \psi + \rho\sigma_1\sigma_2\partial_{\boldsymbol{y}_1}\varphi \partial_{\boldsymbol{y}_2} \psi \right]\,\,d\boldsymbol{y} \\
%         + \int_{\mathbb{R}^2} \left[\lambda_1 (m_1 - \boldsymbol{y}_1)\partial_{\boldsymbol{y}_1}\varphi + \lambda_2 (m_2 - \boldsymbol{y}_2)\partial_{\boldsymbol{y}_2}\varphi \right]\psi \,\,d\boldsymbol{y} = \int_{\mathbb{R}^2} f\psi \,\,d\boldsymbol{y}, \text{ for every } \psi \in C^\infty_c\left(\mathbb{R}^2\right),\\
%         \left\langle \varphi \right\rangle = 0.
%     \end{cases}
% \end{equation}
\end{corollary}
\begin{proof}
Let us notice that
$$
-\left\langle \varphi^{k+1}\mathcal{L}\varphi^{k+1} + \eta \left( \varphi^{k+1} \right)^2 \right\rangle = -\left\langle \varphi^{k+1} f + \varphi^{k+1}\varphi^k \right\rangle.
$$
On the one hand, proceeding in a similar way as we did in the proof of Proposition \ref{prop:PoissonNumericalSolution}, derive
\begin{equation} \label{eq:weakSoln_pt1}
    -\left\langle \varphi^{k+1}\mathcal{L}\varphi^{k+1} + \eta \left( \varphi^{k+1} \right)^2 \right\rangle \geqslant C \left\langle \left|\nabla \varphi^{k+1}\right|^2 \right\rangle - \eta \left\langle \left(\varphi^{k+1} \right)^2\right\rangle. 
\end{equation}
On the other hand,
\begin{equation} \label{eq:weakSoln_pt2}
    -\left\langle \varphi^{k+1} f + \varphi^{k+1}\varphi^k \right\rangle \leqslant C\left( \left\langle \left(\varphi^{k+1} \right)^2\right\rangle + \left\langle \left(\varphi^{k} \right)^2\right\rangle + \left\langle \left(f\right)^2\right\rangle \right) .
\end{equation}
Putting \eqref{eq:weakSoln_pt1} and \eqref{eq:weakSoln_pt2} together, and having in mind that
$$
\sup_k\left\langle \left(\varphi^{k} \right)^2\right\rangle < \infty,
$$
we infer that
$$
\sup_k \left\langle \left|\nabla \varphi^{k+1}\right|^2 \right\rangle < \infty.
$$
Thus, up to a subsequential refinement, $\left\{ \nabla \varphi^k \right\}_k$ converges weakly in the Hilbert space $L^2(\Pi)^2;$ hence, its limit must be equal to $\nabla\varphi$ distributionally. Indeed, let us fix $i \in \left\{1,2\right\}.$ On the one hand, the weak convergence $\nabla\varphi^k \rightharpoonup (\zeta^1,\,\zeta^2)$ in $L^2(\Pi)^2$ implies, for $\psi \in C^\infty_c(\mathbb{R}^2),$ the relation
\begin{equation} \label{eq:WeakDeriv1}
    \int_{\mathbb{R}^2} \partial_{\boldsymbol{y}_i}\varphi^k \psi \,d\boldsymbol{y} = \int_{\mathbb{R}^2} \partial_{\boldsymbol{y}_i}\varphi^k \frac{\psi}{\Pi}\Pi \,d\boldsymbol{y} \rightarrow \int_{\mathbb{R}^2} \zeta^i \frac{\psi}{\Pi}\Pi \,d\boldsymbol{y} = \int_{\mathbb{R}^2} \zeta^i  \psi  \,d\boldsymbol{y}.
\end{equation}
since $\psi/\Pi \in C^\infty_c(\mathbb{R}^2) \subset L^2(\Pi).$ On the other hand, we have
\begin{equation} \label{eq:WeakDeriv2}
    \int_{\mathbb{R}^2}\partial_{\boldsymbol{y}_i}\varphi^k \psi \,d\boldsymbol{y} = - \int_{\mathbb{R}^2}\varphi^k \partial_{\boldsymbol{y}_i}\psi \,d\boldsymbol{y} \rightarrow - \int_{\mathbb{R}^2}\varphi \partial_{\boldsymbol{y}_i}\psi \,d\boldsymbol{y}.
\end{equation}
From \eqref{eq:WeakDeriv1} and \eqref{eq:WeakDeriv2}, we deduce that $ \partial_{\boldsymbol{y}_i} \varphi = \zeta^i$ holds distributionally. In this manner, we have proved our first assertion. 

Next, let us observe that the strong convergence $\varphi^k \rightarrow \varphi$ in $L^2(\Pi),$ alongside the relations $\left\langle \varphi^k \right\rangle = 0$ (valid for every positive integer $k$), allows us to infer that $\left\langle \varphi \right\rangle = 0.$ Finally, with the convergence properties of the sequence $\left\{\varphi^k \right\}$ we have at hand, checking that $\varphi$ solves \eqref{eq:PoissonAvgZero} in the sense of distributions is straightforward --- we omit the details here. 

\end{proof}

%%%%%%%%%%%%%%%%%%%%%%%%%%%%%%%%%%%%%%%%%%%%%%%%%%%%%%%%%%%%%% %%%%%%%%%%%%%%%%%%%% Figures %%%%%%%%%%%%%%%%%%%%%%%%%%%%%%%%% %%%%%%%%%%%%%%%%%%%%%%%%%%%%%%%%%%%%%%%%%%%%%%%%%%%%%%%%%%%%%% 
%\begin{figure}
%    \subfigure[$v_1$]{\includegraphics[width=80mm, scale=0.8]{}}
%    \subfigure[$v_2$]{\includegraphics[width=80mm, scale=0.8]{}}
%    \caption{Caption}
%    \label{fig:my_label}
%\end{figure}
%\FloatBarrier

\end{document}